%% file: main.tex
\documentclass[9pt]{sigplanconf}
\pdfoutput=1
\usepackage{times}
\usepackage{amsmath, amssymb,amsthm}
\usepackage{paralist}
\usepackage{lastpage}
\usepackage{hyperref}
\usepackage{listings}
\usepackage{mathpartir}
\usepackage{utf8}
\usepackage{tikz}
\usepackage{stmaryrd}
\usepackage{xspace}
\usepackage{dsfont}
\usepackage{bbold}
\usepackage{declarations}
\usepackage{balance}
\usetikzlibrary{automata,positioning}

\lstdefinelanguage{z3}{
  keywords={declare, const, assert, Int, Real, and, or, forall, minimize, maxize, check, sat, get, model},
}

\lstdefinelanguage{WHILE}{
  keywords={Lap, while, if, then, else, cons, havoc},
}

\hypersetup{
  hidelinks=true,
  colorlinks=true,
  linkcolor=black,
  citecolor=black,
  urlcolor=black
}

\let\origthelstnumber\thelstnumber
\makeatletter
\newcommand*\Suppressnumber{%
  \lst@AddToHook{OnNewLine}{%
    \let\thelstnumber\relax%
     \advance\c@lstnumber-\@ne\relax%
    }%
}

\newcommand*\Reactivatenumber{%
  \lst@AddToHook{OnNewLine}{%
   \let\thelstnumber\origthelstnumber%
   \advance\c@lstnumber\@ne\relax}%
}
\makeatother

\newcommand\laplace{\ensuremath{\cod{Lap}}\xspace}
\newcommand\uniform{\ensuremath{\cod{Uniform}}_{[0,1]}}
\newcommand\lapm[1]{\laplace~{#1}}

\newcommand\expm[3]{\ensuremath{\cod{Exp}~{#1}~{#2}~{#3}}}

\newcommand\op{\cod{op}}
\newcommand\tylist{\cod{list}}

\newcommand\kout{\cod{return}}
\newcommand\refine{\cod{refine}}
\newcommand\infers{\Join}

\newcommand\leftv[1]{{#1}^{(1)}}
\newcommand\rightv[1]{{#1}^{(2)}}
\newcommand\restrict[3]{{#1}\upharpoonright^{#3}_{#2}}

\newcommand\outcmd[1]{\kout~{#1}}
\newcommand\havoc[1]{\cod{havoc}~{#1}}

\newcommand\prop[1]{\setlength\fboxsep{1pt}\colorbox{gray}{\{\ensuremath{#1}\}}}
\newcommand\invariant[1]{\setlength\fboxsep{1pt}\colorbox{green}{Invariant: \ensuremath{#1}}}
\newcommand\assert[1]{\setlength\fboxsep{1pt}\colorbox{orange}{\ensuremath{#1}}}
\newcommand\instrument[1]{\underline{\ensuremath{#1}}}
\newcommand\annotation[1]{\colorbox{lightgray}{\ensuremath{#1}}}
\newcommand\tT{{\tilde T}}
\newcommand\post{{Φ}}
\newcommand\pre{{\Psi}}
\newcommand\predicate{\mathcal{P}}
\newcommand\inv{{Ι}}
\newcommand\store{{m}}
\newcommand\Store{\mathcal{M}}
\newcommand\distance[1]{\widehat{#1}}
\newcommand\gammadis{\Delta Γ_{\priv}}
\newcommand\joinmem[1]{\oplus_{#1}}
\newcommand\vpriv{{\mathbf{v}_{\priv}}}
\newcommand\vdelta{{\mathbf{v}_{δ}}}
\newcommand\evalexpr[2]{\trans{#1}_{#2}}
\newcommand\real{{r}}
\newcommand\dexpr{{\mathbb{d}}}
\newcommand\htriple[3]{\{{#1}\}{#2}\{{#3}\}}
\newcommand\transform{\rightharpoonup}
\newcommand\samples{π}
\newcommand\priv{\ensuremath{\epsilon}}

\newcommand\dist{\mathbf{Dist}}
\newcommand\extmemfull[3]{{#1}\uplus({#2})} 
\newcommand\extmem{{{\store}\uplus(\priv)}} 
\newcommand\support{\cod{support}}
\newcommand\unitop{\cod{unit}}
\newcommand\bindop{\cod{bind}}
\newcommand\langzero{LightDP$_0$\xspace}
\newcommand\lang{LightDP\xspace}
\newcommand\basety{\mathcal{B}}
\newcommand\dgdist{\ensuremath{\mathds{1}}\xspace}

\newcommand\funsigfour[4]{
\makebox[1.1cm]{\textbf{function}} \textsc{#1}~(#2) \par
\makebox[1.1cm]{}                  \textbf{returns}~{#3}\par
\makebox[1.1cm]{}                  \textbf{precondition}~{#4}\par}

\newif\ifblinded\blindedfalse
\newif\ifreport\reporttrue

\title{\lang: Towards Automating Differential Privacy Proofs}
\ifblinded
\authorinfo{Anonymized for review}{}{}
\else
\authorinfo{Danfeng Zhang \and Daniel Kifer}{Department of Computer Science and Engineering \\ Penn State University\\ University Park, PA  United States}{\{zhang,dkifer@cse.psu.edu\}}
\fi
\begin{document}

\CopyrightYear{2017} 
\publicationrights{transferred}
\conferenceinfo{POPL '17,}{January 18-20, 2017, Paris, France}
\copyrightdata{978-1-4503-4660-3/17/01}
\copyrightdoi{3009837.3009884}

\maketitle

\begin{abstract}
\input{abstract}
\end{abstract}

\category{D.3.1}{Programming Languages}
{Formal Definitions and Theory}
\category{D.2.4}{Software Engineering}
{Software/Program Verification}
\category{F.3.1}{Logics and Meanings of Programs}
{Specifying and Verifying and Reasoning about Programs.}

\keywords
Differential privacy; dependent types; type inference; 

\section{Introduction}\label{sec:intro}
\input{intro}

\section{Preliminaries and Illustrating Example}
\label{sec:preliminaries}
\label{sec:illustrative}

\input{preliminaries}

\input{illustrative}

\section{\lang: A Language for Algorithm Design}\label{sec:langzero}

We first introduce a simple imperative language, \langzero, for designing and
verifying privacy-preserving algorithms. This language is equipped with a
dependent type system that enables formal verification of sophisticated
algorithms where the composition theorem falls short. In this section, we
assume all type annotations are provided by a programmer. We will remove this
restriction and enable type inference in Section~\ref{sec:fulllang}.

\subsection{Syntax}
\label{sec:syntax}
The language syntax is given in Figure~\ref{fig:syntax}. \langzero is mostly a
standard imperative language except for the following features.

\begin{figure}
$
\begin{array}{l@{\ }l@{\ }c@{\ }l}
\text{Reals} & r &\in &\mathbb{R} \\
\text{Booleans} & b &\in &\{\true, \false\} \\
\text{Vars} & x  &\in &\Vars \\
\text{Rand Vars} & η  &\in &Η \\
\text{Linear Ops} & \oplus &::= &+ \mid - \\
\text{Other Ops} & \otimes &::= &\times \mid / \\
\text{Comparators} & \odot &::= &< \mid > \mid = \mid ≤ \mid ≥ \\
\text{Rand Exps}   & g &::=\; &\lapm{\real} \\ 
\text{Expressions} & e &::=\; &\real \mid b \mid x \mid η \mid e_1\oplus e_2 \mid e_1\otimes e_2 \mid e_1 \odot e_2 \mid \\
                   & & & \neg e \mid e_1::e_2 \mid e_1[e_2] \mid e_1?e_2:e_3 \\
\text{Commands} & c &::=\; &\skipcmd \mid x := e \mid η := g \mid c_1;c_2 \mid \outcmd{e} \mid \\ 
                  & & &\ifcmd{e}{c_1}{c_2} \mid \whilecmd{e}{c} \\
\text{Distances}& \dexpr &::=\; & \real \mid x \mid η \mid \dexpr_1 \oplus \dexpr_2 \mid \dexpr_1 \otimes \dexpr_2 \mid \dexpr_1 \odot \dexpr_2 ?\dexpr_3:\dexpr_4 \\
\text{Types} & τ &::=\; &\tyreal_{\dexpr} \mid \tyreal_* \mid \bool \mid \tylist~τ \mid τ_1→τ_2 
\end{array}
$
\caption{\langzero: language syntax.}
\label{fig:syntax}
\end{figure}

\paragraph{Random expressions} 
Probabilistic reasoning is essential in privacy-preserving algorithms. We use
$g$ to represent a random expression. Since \langzero follows a modular design
where new randomness expression can be added easily, we only consider the most
interesting random expression, $\lapm r$, for now. Semantically, $\lapm r$
draws one sample from the Laplace distribution, with mean zero and a scale
factor $r$. We will discuss other random expressions in Section~\ref{sec:categorical}.

Each random expression $g$ can be assigned to a random variable $η$, written
as $η := g$. We distinguish random variables (Η) from normal variables
($\Vars$) for technical reasons explained in Section~\ref{sec:typing}. 
Notice that although the syntax restricts the distribution scale parameters to
be a constant, its mean can be an arbitrary expression $e$, via the legit
expression $e+η$, where $η$ is sampled from a distribution with mean zero.

\paragraph{List operations} 
Sophisticated algorithms usually make multiple queries to a database and
produce multiple outputs during that process. Rather than reasoning about the
privacy cost associated with each query in isolation and total the privacy
costs using the composition theorem, \langzero enables more precise reasoning
via built-in list type operations: $e_1::e_2$ appends
the element $e_1$ to a list $e_2$; $e_1[e_2]$ gets the $e_2$-th element in list
$e_1$, assuming $e_2$ is bound by the length of $e_1$. We also assume a list
variable is initialized to an empty list.

\paragraph{Types with distances}
Each type $τ$ has the form of $\basety_\dexpr$. Here, $\basety$ is a base type,
such as $\tyreal$ (numeric type), $\bool$ (Boolean), or an application of a type
constructor (e.g., $\tylist$) to another type, or a function ($τ_1→τ_2$). $\dexpr$ is a numeric  
expression that (semantically) specifies the exact distance of the values
stored in a variable in two related executions. In particular, a distance
expression is a numeric expression in the language, as specified in
Figure~\ref{fig:syntax}, where $\dexpr_1 \odot \dexpr_2 ?\dexpr_3:\dexpr_4$
evaluates to $\dexpr_3$ when the comparison evaluates to $\true$, and
$\dexpr_4$ otherwise.

Since non-numeric types $\bool$, $\tylist~τ$ and $τ_1→τ_2$ cannot be associated with any
numeric distance, those types are syntactic sugars for
$\bool_0$, $(\tylist~τ)_0$ and $(τ_1→τ_2)_0$ respectively. Notice that elements in a list of
type $(\tylist~τ)_0$ (e.g., parameter $q$ in Figure~\ref{fig:sparsevector}) may
still have different elements in two related executions, since the difference
of elements is specified by type $τ$. The subscript $0$ here simply
restricts the list size in two related executions.

\paragraph{Star type} 
\langzero also supports sum types, written as $\basety_*$, a syntactic sugar
 of a Sigma type. More specifically, a variable $x$ with type $\tyreal_*$
is desugared as $x:\Sigma_{(\distance{x}:\tyreal_0)}~\tyreal_{\distance{x}}$,
where $\distance{x}$ is a distinguished variable invisible in the source code, but
can be reasoned about and manipulated by the type system. Hiding the first
component of a Sigma type simplifies verification
(Section~\ref{sec:soundness}).

The parameter $q$ in Figure~\ref{fig:sparsevector} is one example where the
star type is useful. Moreover, the star type enables reasoning about
dependencies that cannot be captured otherwise by a distance expression.
Consider the $\priv$-differentially private Partial Sum algorithm in Figure~\ref{fig:partialsum}. It
 implements an immediate solution to answering the sum
of a query list in a privacy preserving manner\footnote{We use this trivial
algorithm here for its simplicity. We will analyze a more sophisticated version
in Section~\ref{sec:casestudy}.}: it aggregates the accurate partial sum in a
loop, and releases a noisy sum using the Laplace mechanism. The precondition
specifies the adjacency assumption: at most one query answer may differ by at
most $b$.

\begin{figure}
\small

\algrule
\funsigfour{PartialSum}{$\priv,b,\cod{size}:\tyreal_0; q:{\tylist~\tyreal_*} $}{($\cod{out}:\tyreal_0$)}{$∀i≥0.~\distance{q}[i]≤b ∧$ \par \makebox[1.1cm]{} 
                                        $∀i≥0.~\distance{q}[i]>0 ⇒ (∀j>i.~\distance{q}[j]=0)$}
\algrule
\begin{lstlisting}[frame=none]
$\annotation{\cod{sum}:\tyreal_*;~\cod{i}:\tyreal_0;η:\tyreal_{-\distance{\cod{sum}}}}$
sum := 0; i := 0;
while (i < size)
  sum := sum+q[i];
  i := i+1;
$η$ = $\lapm{b/\priv}$;
out := sum + $η$;
\end{lstlisting}
\caption{An \priv-differentially private algorithm for summing over a query list.}
\label{fig:partialsum}
\end{figure}

In this algorithm, the distance of variable $\cod{sum}$ changes in each
iteration. Hence, the accurate type for $\cod{sum}$ is $\tyreal_{\sum_{j=0}^i
\distance{q}[j]}$. However, with the goal of keeping type system as
light-weight as possible, we assign $\cod{sum}$ to a star type. The type system
will reason about and manipulate distance component $\distance{\cod{sum}}$ in a
sound way (Section~\ref{sec:typing}).

\subsection{Semantics}
\label{sec:semantics}

\begin{figure}
\begin{align*}
\trans{\skipcmd}_\store &= \unitop~m \\
\trans{x:=e}_\store &= \unitop~(\subst{m}{x}{\evalexpr{e}{m}}) \\
\trans{η:=g}_\store &= \bindop~\evalexpr{g}{}~(λv.~\unitop~\subst{m}{η}{v}) \\
\trans{c_1;c_2}_\store &= \bindop~(\evalexpr{c_1}{\store})~\trans{c_2}  \\
\trans{\ifcmd{e}{c_1}{c_2}}_\store &= 
      \begin{cases} 
      \trans{c_1}_\store &\mbox{if } \evalexpr{e}{\store} = \true \\
      \trans{c_2}_\store &\mbox{if } \evalexpr{e}{\store} = \false \\
      \end{cases} \\
\trans{\whilecmd{e}{c}}_\store &= w^*~m\\
\text{where } w^*              &= fix (λf.~λm. \cod{if}~\evalexpr{e}{m}=\true\\
                               &\qquad \cod{then}~(\bindop~\evalexpr{c}{m}~f)~\cod{else}~{(\unitop~m)}) \\
\trans{c;\outcmd{e}}_\store &= \bindop~(\evalexpr{c}{\store})~(λm'.~\unitop~\evalexpr{e}{m'})
\end{align*}
\caption{\langzero: language semantics.}
\label{fig:semantics}
\end{figure}

The denotational semantics of the probabilistic language is defined as a mapping
from initial memory to a distribution on (possible) final outputs. Formally,
let $\Store$ be a set of memory states where each memory state $\store \in
\Store$ is an assignment of all (normal and random) variables ($\Vars∪Η$) to
values.
First, an expression $e$ of base type $\basety$ is interpreted as a function
$\trans{e}: \store → \trans{\basety}$, where $\trans{\basety}$ represents the
set of values belonging to the base type $\basety$. We omit expression
semantics since it is mostly standard\footnote{The reals in \lang only come
from sampling (or, the havoc command, which mimics sampling). We assume the
sample space is either finite or countable.}.

A random expression $g$ is interpreted as a distribution on real values. Hence,
$\trans{g}:\dist(\trans{\tyreal})$. Moreover, a command $c$ is interpreted as a
function $\trans{c}:\Store → \dist(\Store)$. For brevity, we write
$\evalexpr{e}{\store}$ and $\evalexpr{c}{\store}$ instead of
$\trans{e}(\store)$ and $\trans{c}(\store)$ hereafter.
Figure~\ref{fig:semantics} provides the semantics of commands, where functions
$\unitop$ and $\bindop$ are defined in Section~\ref{sec:dist}. This semantics
corresponds directly to a semantics given by \citet{Kozen81}, which interprets
programs as continuous linear operators on measures. 

Finally, we assume all programs have the form $(c;\outcmd{e})$ where $c$ does
not contain return statements. A \langzero program is interpreted as a function
$\store → \dist{\trans{\basety}}$, defined in Figure~\ref{fig:semantics}, where
$\basety$ is the type of expression returned ($e$).

\subsection{Typing Rules and Target Language}
\label{sec:typing}

\begin{figure}
\framebox{Typing rules for expressions.}
\begin{mathpar}
\inferrule*[right=(T-Num)]{ }{ Γ  \proves \real : \tyreal_0}
\and
\inferrule*[right=(T-Boolean)]{ }{ Γ  \proves b : \bool}
\and
\inferrule*[right=(T-Var)]{ }{ Γ, x:\basety_{\dexpr}  \proves x : \basety_{\dexpr}}
\quad
\inferrule*[right=(T-VarStar)]{ }{ Γ, x:\basety_*  \proves x : \basety_{\distance{x}}}
\and
\inferrule*[right=(T-OPlus)]{Γ\proves e_1: \tyreal_{\dexpr_1} \quad Γ\proves e_2 : \tyreal_{\dexpr_2}}{ Γ \proves e_1\oplus e_2 : \tyreal_{\dexpr_1\oplus \dexpr_2}}
\and
\inferrule*[right=(T-OTimes)]{Γ\proves e_1: \tyreal_0 \quad Γ\proves e_2 : \tyreal_0}{Γ \proves e_1\otimes e_2 : \tyreal_0}
\and
\inferrule*[right=(T-ODot)]{\inferrule*{}{Γ\proves e_1: \tyreal_{\dexpr_1} \\\\ Γ\proves e_2 : \tyreal_{\dexpr_2}} \qquad
                             \inferrule*{}{\pre ⇒ (e_1\odot e_2 \\\\ ⇔ (e_1\!+\!\dexpr_1)\odot (e_2\!+\!\dexpr_2))}}
                             { Γ \proves e_1\odot e_2 : \bool}
\and
\inferrule*[right=(T-Neg)]{Γ\proves e: \bool}{ Γ \proves \neg e : \bool}
\quad
\inferrule*[right=(T-Cons)]{Γ\proves e_1:τ \quad Γ\proves e_2:\tylist~τ}
               {Γ \proves e_1::e_2 : \tylist~τ}
\and
\inferrule*[right=(T-Index)]{Γ\proves e_1:\tylist~τ\quad Γ\proves e_2:\tyreal_0}{Γ \proves e_1[e_2]:τ}
\and
\inferrule*[right=(T-Select)]{Γ\proves e_1:\bool \quad Γ\proves e_2:τ \quad Γ\proves e_3:τ}{Γ \proves e_1?e_2:e_3:τ}
\end{mathpar}

\framebox{Typing rules for commands}
\begin{mathpar}
\inferrule*[right=(T-Skip)]{ }{Γ \proves \skipcmd \transform \skipcmd}
\and
\inferrule*[right=(T-Asgn)]{ Γ\proves e: τ \quad Γ\proves x: \basety_\dexpr \quad τ=\basety_\dexpr }
                           {Γ \proves x := e \transform x := e}
\and
\inferrule*[right=(T-AsgnStar)]{ Γ\proves x: \basety_{\distance{x}} \and Γ\proves e:\basety_\dexpr }
                           {Γ \proves x := e \transform x := e; \distance{x} := \dexpr;}
\and
\inferrule*[right=(T-Seq)]{Γ \proves c_1\transform c_1' \quad Γ \proves c_2\transform c_2'}{Γ \proves c_1;c_2  \transform c_1';c_2'}
\and
\inferrule*[right=(T-Return)]{Γ \proves e:\tyreal_0 {\text\quad or\quad } Γ \proves e:\bool_0}
                             {Γ \proves \outcmd{e}\transform \outcmd{e}}
\\
\inferrule*[right=(T-If)]{Γ\proves e : \bool \quad Γ \proves c_i\transform c_i' \text{ where } i\in\{1,2\}}
{ Γ \proves \ifcmd{e}{c_1}{c_2} \transform \ifcmd{e}{c_1'}{c_2'}}
\\
\inferrule*[right=(T-While)]{Γ\proves e:\bool \and Γ\proves c\transform c'}{Γ \proves \whilecmd{e}{c} \transform \whilecmd{e}{c'}}
\end{mathpar}

\framebox{Typing rules for random assignments}
\begin{mathpar}
\inferrule*[right=(T-Laplace)]{ Γ(η) = \tyreal_{\dexpr}}
                           {Γ \proves η := \lapm~\real \transform \havoc{η}; \vpriv = \vpriv+|\dexpr|/r;}
\end{mathpar}

\caption{Typing rules. $\pre$ is an invariant that holds throughout program
execution.}
\label{fig:typing}
\end{figure}

We assume a typing environment $Γ$ that tracks the type of each variable
(including random variable). For now, we assume a type annotation is provided for
each variable (i.e., $\dom(Γ)=\Vars∪Η$), but we will remove this restriction in
Section~\ref{sec:fulllang}. The typing rules are formalized in
Figure~\ref{fig:typing}. Since all typing rules share a global invariant $\pre$
(e.g., the precondition in Figure~\ref{fig:sparsevector}), typing rules do not
propagate $\pre$ for brevity. We also write $Γ(x)=\dexpr$ for
$∃\basety.~Γ(x)=\basety_\dexpr$ when the context is clear.

\paragraph{Expressions}
For expressions, each rule has the form of $Γ\proves e:τ$, meaning that the
expression $e$ has type $τ$ under the environment $Γ$. Rule~\ruleref{T-OPlus}
precisely tracks the distance of linear operations (e.g., $+$ and $-$), while
rule~\ruleref{T-OTimes} makes a conservative assumption that other numerical
operations take identical parameters. It is completely possible to refine
rule~\ruleref{T-OTimes}) (e.g., by following the sensitivity analysis proposed
by~\citet{Fuzz,DFuzz}) to improve precision, however, we leave that as future
work since it is largely orthogonal.

Rule~\ruleref{T-VarStar} applies when variable $x$ has a star type. This rule
unpacks the corresponding pair with Sigma type and makes $\distance{x}$
explicit in the type system.

The most interesting and novel rule is \ruleref{T-ODot}. It type-checks a
comparison of two real expressions by generating a constraint:
\[\pre ⇒ (e_1\odot e_2 ⇔ (e_1\!+\!\dexpr_1)\odot (e_2\!+\!\dexpr_2))\]

Intuitively, this constraint requires that in two related executions, the
Boolean value of $e_1\odot e_2$ must be identical since the distances of $e_1$
and $e_2$ are specified by $\dexpr_1$ and $\dexpr_2$ respectively.
For example, consider the branch condition $q[i]+η_2≥\tT$ in
Figure~\ref{fig:sparsevector}.  Rule~\ruleref{T-ODot} first checks types for
subexpressions: 
\[Γ\proves q[i]+η_2: \tyreal_{\distance{q}[i]+(q[i]+η_2≥\tT?2:0)} \text{ and } Γ\proves \tT: \tyreal_1\]
Then the following constraint is generated (free variables in the generated
constraint are universally quantified):
\vspace{-2ex}
\begin{multline*}
∀q_i,\distance{q_i},η_2,\tT\in \mathbb{R}.~(-1≤\distance{q_i}≤1)⇒ \\
q_i+η_2≥\tT ⇔
q_i+η_2+\distance{q_i}+(q_i+η_2≥\tT?2:0) ≥ \tT+1
\end{multline*}

This proof obligation captures a subtle yet important property of the Sparse
Vector method: given the randomness alignment as specified by $Γ$, two related
executions must take the same branch (hence, produce the same output).  This
proof obligation can easily be discharged by an external SMT solver, such as
Z3~\cite{z3}.

\paragraph{Target language}
The typing rules for a command have the form of $Γ\proves c\transform c'$, where
$c$ is the original program being verified, and $c'$ is the transformed program
in the target language defined in Figure~\ref{fig:targetlang}. The target
language is mostly identical to the original one, except for two significant
differences:
\begin{inparaenum}[1)]
\item the target language involves a distinguished variables $\vpriv$ to
explicitly track the privacy cost in the original program; 
\item the target language removes probabilistic expressions, and introduces a
new nondeterministic command $(\havoc{x})$, which sets variable $x$ to an
arbitrary value upon execution.
\end{inparaenum}
Hence, the target language is nonprobabilistic.

Due to nondeterminism, the denotational semantics interprets a command $c$ in
the target language as a function $\trans{c}: \Store→\mathcal{P}(\Store)$. For
example, the semantics of the $\cod{havoc}$ command is defined as follows:
\[\evalexpr{\havoc{x}}{m} = ∪_{r\in \mathbb{R}}~\{\subst{m}{x}{r}\} \]
Other commands have a standard semantics, hence their semantics are included in
\ifreport
the appendix.
\else
the full version of this paper~\cite{report}. 
\fi
Note that for simplicity, we abuse the notation $\trans{c}$ to denote the
semantics of both the source language and target language. However, its meaning
is unambiguous in the context of a memory $m$: when $\vpriv \not\in \dom(m)$,
$\evalexpr{c}{m}$ denotes a distribution; otherwise, $\evalexpr{c}{m}$ denotes
a set.

\begin{figure}
$
\begin{array}{l@{\ }l@{\ }c@{\ }l}
\text{Vars} & x &\in &\Vars ∪ Η ∪ \{\vpriv\} \\
\text{Statements} & c &::=\; &\skipcmd \mid x := e \mid \havoc{x} \mid c_1;c_2 \mid \outcmd{e} \\ 
                  & & & \ifcmd{e}{c_1}{c_2} \mid \whilecmd{e}{c} \\
\end{array}
$
\caption{Target language syntax. The omitted parts are identical to the source
language defined in Figure~\ref{fig:syntax}.}
\label{fig:targetlang}
\end{figure}

\paragraph{Commands}
Informally, if $Γ\proves c \transform c'$, and the distinguished variable
$\vpriv$ in $c'$ is bounded by some constant $\priv'$ in all possible
executions, then program $c$ is $\priv'$-differentially private. Here, we discuss important
typing rules to enforce this property.  We will formalize this soundness
property and sketch a proof in Section~\ref{sec:soundness}.

For an assignment $x:=e$, rule~\ruleref{T-Asgn} \emph{synthesizes} the types of
$x$ and $e$, and \emph{checks} that their types are equivalent (i.e., both the
base type and distance are equivalent). However,
rule~\ruleref{T-AsgnStar} instruments the original program so that the typing
invariant (i.e., the distance of $x$ is exactly $\distance{x}$) is maintained
after the assignment. Consider line $4$ in Figure~\ref{fig:partialsum}.
Rule~\ruleref{T-AsgnStar} first checks subexpressions: $Γ\proves
\cod{sum}+q[i]:\distance{\cod{sum}}+\distance{q}[i]$. Hence, the transformed
program is $(\cod{sum} := \cod{sum}+q[i];\distance{\cod{sum}} :=
\distance{\cod{sum}}+\distance{q}[i]$), which correctly maintains the typing
invariant after the assignment.

\ruleref{T-Return} checks that the returned value is indistinguishable in two
related executions. Both~\ruleref{T-If} and~\ruleref{T-While} check that
two related executions must follow the same control flow.

\paragraph{Laplace mechanism} Intuitively, \ruleref{T-Laplace} assigns a
polymorphic type to the random source $\lapm{r}$. In other words, for any
distance $\dexpr$ of random variable $η$, we can instantiate the type of
$\lapm{r}$ to be $\tyreal_\dexpr$, though with a privacy cost of $|\dexpr|/r$.
Moreover, the transformed program sets $η$ to a nondeterministic value
($\havoc{η}$), since any real value can be sampled from the Laplace
distribution.

Consider line $1$ in Figure~\ref{fig:sparsevector}. \ruleref{T-Laplace}
transforms this line to $(\havoc{η_1};\vpriv:=\vpriv+\priv/2)$ since
$Γ(η_1)=\tyreal_1$. Informally, the transformation says that by paying a privacy
cost of $\priv/2$, $\lapm~(2/\priv)$ ensures that $η_1$ has a distance of one
in two related executions. 

Moreover, consider line $6$ in Figure~\ref{fig:partialsum}. \ruleref{T-Laplace}
transforms this line to $(\havoc{η};\vpriv:=\vpriv+|\distance{\cod{sum}}|\priv/b)$
since $Γ(η)=-\distance{\cod{sum}}$. Informally, the transformation says that by paying a
privacy cost, $η$ has a distance of $-\distance{\cod{sum}}$ in two related
executions. Hence, we can cancel out the distance of $\cod{sum}$ in line 7.

\paragraph{Dependent types and imperative programming}
Mutable states in imperative programming brings subtleties that are not
foreseen in the standard theory of dependent types~\cite{intuitionistic}.
Consider a variable $x$ with type $\tyreal_y$ where $y$ is initialized to zero.
The type of $x$ establishes an invariant on its values $v_1,v_2$ under two
executions: $\leftv{v}=\rightv{v}$. However, if we update the value of $y$ to
$1$, the invariant changes to $\leftv{v}+1=\rightv{v}$, but the values of $x$
in two executions remains unchanged. Hence, the type invariant is broken.

To address this issue, we assume the following assumptions are checked before
type checking. First, for each normal variable $x\in \Vars$ such that
$Γ(x)=\basety_\dexpr$, all free variables in $\dexpr$ are immutable. For
example, each normal variables in Figure~\ref{fig:sparsevector} has a constant
distance in its type. Note that by language syntax, this restriction does not
apply to variables with a star type.
Second, a random variable $η\in Η$ may depend on mutable variables. However, we
assume that it has only \emph{one use} other than the definition, and the
definition of $η$ is adjacent to its use. Hence, each variable that $η$ depends
on appears immutable between $η$'s definition and use~\footnote{The one-use
assumption may appear restrictive at first glance, but when $η$'s type has no
dependency on mutable variables, we can always store $η$ to a normal variable
to circumvent this restriction. When $η$'s type depends on a mutable variable
and multiple uses of $η$ are needed, we can store $η$'s value to a normal
variable with star type, whose distance counterpart is manipulated and reasoned
about by the type system.}.
 
\if0
Note that when $S=∅$, the difference is 0 by definition. so
$\distance{\priv}(μ_1,μ_2)≥0$ for any $\priv, μ_1, μ_2$.

By the definition of \priv-distance, a probabilistic program $c$ is
\priv-differentially private with respect to $\priv>0$ and a relation $\pre$ on
two initial stores $\store_1$ and $\store_2$ if
\[\distance{\priv}(\trans{c}_{\store_1},\trans{c}_{\store_2})≤0\]

The proof of our main theorem relies on a lifting operator that turns a
relation on stores into a relation on distributions over stores. Intuitively,
the lifting relates two distributions $μ_1\in D(A)$ and $μ_2\in D(B)$ whenever
there exists a distribution $μ$ over $A\times B$ whose support is contained in
relation $R$ and whose first and second projections are at most at
\priv-distance δ of $μ_1$ and $μ_2$ respectively.

The usefulness of the lifting of relations over stores is established by the
following theorem.

\begin{theorem}
Let $μ_1\in D(A), μ_2\in D(B)$, and $R⊆A\times B$. Then for any two functions
$f_1:A→[0,1]$ and $f_2:B→[0,1]$,
\[μ_1 \sim^{\priv,δ}_R μ_2 ⇒ \distance{\priv}(μ_1 f_1, μ_2 f_2) ≤ δ \]
\end{theorem}

In particular, if $A=B$ and $R$ is the equality relation, then 
\[μ_1 \sim^{\priv,δ}_= μ_2 ⇒ \distance{\priv}(μ_1, μ_2) ≤ δ \]

For Laplace mechanism, the following lemma is proved in previous work:
\begin{lemma}
\[\distance{|\leftv{x}-\rightv x|\priv}(\lapm{x}{\priv}, \lapm{\rightv x}{\priv})≤0\]
\end{lemma}

The theorem above provides an tyrealrpretation of lifting in terms of
\priv-distance. Next, we state the main theorem of this work. 

Also useful is the following lemma (Lemma~8 in~\cite{Barthe14}):
\begin{lemma}
\label{lem:distribution}
For all distribution expressions $μ_1,μ_2$, if 
\[\distance{\priv}(\evalexpr{μ_1}{\store},\evalexpr{μ_2}{\rightv \store})≤δ\]
then
\[(\evalexpr{x:=μ_1}{\store})Q_{\priv,δ}(\evalexpr{x:=μ_2}{\rightv \store})\]
where $Q=\{\store_1,\store_2|\store_1(x)=\store_2(x)\}$.
\end{lemma}

\begin{theorem}
if we have $\proves c : \pre ⇒ \post ∧ \priv ≤ R$, then program $c$ is \priv-private.
\end{theorem}

This theorem follows the Lemma below.

\begin{lemma}
if we have $\proves c : \pre ⇒ \post ∧ \priv ≤ R$, then 
\[∀\store_1, \store_2~.~\store_1 \pre \store_2 ⇒ (\trans{c}_{\store_1}) \sim^{\priv,δ}_{\post} (\trans{c}_{\store_2})\]
\begin{proof}
We introduce some notations to aid the proof first.

For disjoint stores $\store_1$ and $\store_2$ and real value \priv, $\store_1
\joinmem{\priv} \store_2$ denotes the joint memory $\store$ such that
$\store(x)=\store_1(x)$ for all variables defined in $\store_1$,
$\store(x)=\store_2(x)$ for all variables defined in $\store_2$, and
$\store(\vpriv)=\priv$.

Given a relation $R$ on $\Store \times \Store$, let $\hat{R}_\priv$ be the set
$\{\store_1 \joinmem{\priv'} \store_2~|~(\store_1,\store_2)\in R ∧
\priv'≤\priv\}$. The proof follows by showing the following lemma:
\begin{align*}
                & (∀\store.~\store\in\hat{R}_\priv ⇒ ∀\store'.~\store'\in (\evalexpr{c}{\store}) ⇒ \store'\in\hat{Q}_{\priv'}) \\
\Longrightarrow &\  \\
                & ∀\store_1,\store_2.~(\store_1,\store_2)\in R ⇒ (\evalexpr{c}{\store_1})~Q_{\priv'-\priv}(\evalexpr{c}{\store_2})
\end{align*}

Most proof just follow the proof in previous work. The tyrealresting case is for the added case
\begin{mathpar}
\inferrule*[right=(T-If)]{∃R_1,R_2.~\inferrule*{}{
                           (\subst{\pre∧e}{x}{x+R_1}⇒\rightv e) ∧ (\rightv e ⇒ \subst{\pre∧e}{x}{x-R_1}) \\
                           (\subst{\pre∧\neg e}{x}{x+R_2}⇒\neg \rightv e) ∧ (\neg \rightv e ⇒ \subst{\pre∧\neg e}{x}{x-R_2})}}
{ Γ \proves \ifcmd{e}{(o:=\top)}{(o:=\bot)} : \pre ⇒ \leftv{o}=\rightv o} 
\end{mathpar}

By the lemma in previous work, we need to prove that
$\distance{|\priv|}{(\evalexpr{o}{\leftv \store},\evalexpr{o}{\rightv \store'})}
≤ 0$. That is, for all $r$,
\[P(\evalexpr{o}{\store}=r)-2^{|\priv|}P(\evalexpr{o}{\store'}=r) ≤ 0 ∧ P(\evalexpr{o}{\store'}=r)-2^{|\priv|}P(\evalexpr{o}{\store}=r) ≤ 0\]
which is equivalent to prove 
\[∀b\in\{\true,\false\}.~P(\evalexpr{e}{\leftv \store}=b)/P(\evalexpr{e}{\rightv \store}=b) ≤ 2^{|\priv|} ∧ 
                         P(\evalexpr{e}{\rightv \store}=b)/P(\evalexpr{e}{\leftv \store}=b) ≤ 2^{|\priv|}\]

Here, we prove the case when $b=\true$. The other case can be proven in a
similar way. 

When $r=\true$, let $\store_u=\subst{\store}{x}{x+R}$ and
$\store_l=\subst{\store}{x}{x-R}$. We have 
\[P(\evalexpr{e}{\leftv \store})≤2^{R_1\times |\priv|} P(\evalexpr{e}{\store_u})\] and 
\[P(\evalexpr{e}{\store_l})≤2^{R_1\times |\priv|} P(\evalexpr{e}{\leftv \store})\] 
by Laplace mechanism. Since $\evalexpr{e}{\store_u} ⇒ \evalexpr{e}{\rightv
\store}$ by hypothesis, $P(\evalexpr{e}{\store_u})≤P(\evalexpr{e}{\rightv
\store})$. Hence, $P(\evalexpr{e}{\leftv \store}) ≤ 2^{R_1\times |\priv|} P(\evalexpr{e}{\rightv \store_u})
≤ 2^{R_1\times |\priv|} P(\evalexpr{e}{\rightv \store})$.
Similarly, since $\evalexpr{e}{\rightv \store} ⇒ \evalexpr{e}{\store_l}$ by
hypothesis, $P(\evalexpr{e}{\rightv \store})≤ P(\evalexpr{e}{\store_l})$.
Hence, $P(\evalexpr{e}{\rightv \store}) ≤ P(\evalexpr{e}{\store_l}) ≤
2^{R_1\times |\priv|} P(\evalexpr{e}{\leftv \store})$.

\end{proof}
\end{lemma}
\fi

\section{Soundness}
\label{sec:soundness}

The type system in Section~\ref{sec:typing} enforces a fundamental property: if
$Γ\proves c\transform c'$ and $\vpriv$ in $c'$ is bounded by some constant
$\priv$, then the original program being verified is $\priv$-differentially private. 

To formalize and prove this soundness property, we first notice that a typing
environment $Γ$ defines a relation on two memories, since $Γ$
specifies the exact distance of each variable:
\begin{definition}[Γ-Relation]
\label{def:relation} Two memories $\store_1$ and $\store_2$ are related by a
typing environment $Γ$, written $\store_1~Γ~\store_2$, iff \[∀x\in
\Vars.~\store_1(x)+\evalexpr{\dexpr_x}{\store_1}=\store_2(x),\text{ where } Γ\proves x:\basety_{\dexpr_x}\]
\end{definition} 
Note that since $Γ(x)$ might be a dependent type, the definition needs to
evaluate the distance of $x$ ($\dexpr_x$) under $m_1$.

By the definition above, $Γ$ is a function since for any memory $m$, the
distance for each variable in the related memory of $m$ is a constant.  Hence,
we also write $Γ(m)$ to represent the unique $m'$ such that $m~Γ~m'$.
Moreover, given a set of distinct memories $S⊆\Store$, we define $Γ~S\defn
\{Γ(m) \mid m\in S\}$. Note that by definition, $Γ~S$ is also a set of distinct
memories (hence, $\emph{not}$ a multiset).
Furthermore, we assume that $Γ$ is an injective function. We make
this assumption explicit by the following definition. 

\begin{definition}[Well-Formed Γ-relation]
A typing environment $Γ$ is well-formed, written $\proves Γ$, iff $Γ$ is an
injective function. 
\end{definition}

Checking  well-formedness of the  Γ-relation is straightforward. Intuitively,
$Γ$ is well-formed when there is no ``circular'' dependency, while more careful
analysis is needed for circular dependencies. Consider the Sparse Vector method
in Figure~\ref{fig:sparsevector} and any $m_1$ and $m_2$ such that
$Γ~m_1=Γ~m_2=m$ for some $m$.
For a variable $y$ with a constant distance $\cod{v}$ (e.g.,
$\tT,η_1,\distance{q}[i]$), we have $m_1(y)=m(y)-\cod{v}=m_2(y)$. So $m_1$ and
$m_2$ must agree on those variables.  Then for any variable $z$ that depends on
variables that $m_1$ and $m_2$ already agree on, the distance of $z$ must be
identical in $m_1$ and $m_2$; hence, $m_1(z)=m_2(z)$. For the circular
dependency on variable $η_2$ (whose distance depends $η_2$), consider
$\cod{t}=\evalexpr{\tT-q[i]}{m_1}=\evalexpr{\tT-q[i]}{m_2}$. The mapping for
$η_2$ is $v\mapsto v+2$ when $v>\cod{t}$ and $v\mapsto v$ otherwise. Since this
mapping is strictly monotonic, it is injective. 

For differential privacy, we are interested in the relationship between two
\emph{memory distributions}. Given a typing environment $Γ$ and constant
$\priv$, we define the $(Γ,\priv)$ distance, written $\gammadis$, of two memory
distributions:
\begin{definition}[$Γ_\priv$-distance]
The $Γ_{\priv}$-distance of two distributions $μ_1, μ_2\in \dist{(\Store)}$, written
$\gammadis(μ_1,μ_2)$, is defined as: 
\[\gammadis(μ_1,μ_2)\defn \max_{S⊆\Store} (μ_1(S)-\exp(\priv) μ_2(Γ~S))\] 
\end{definition}

Note that when $S=∅$, the distance is 0 by definition. So
$\gammadis(μ_1,μ_2)≥0$ for any $\priv, μ_1, μ_2$. 

The soundness theorem connects the ``privacy cost'' of the probabilistic
program to the distinguished variable $\vpriv$ in the transformed
nonprobabilistic program. In order to formalize the connection,
we first extend memory in the source language to include $\vpriv$:

\begin{definition}
For any memory $m$ and constant $\priv$, there is an extension of $m$,
written $\extmem$, so that 
\begin{align*}
  & ∀x\in \dom(m).~\extmem(x) =m(x) \\
∧~& \extmem(\vpriv)=\priv
\end{align*}
\end{definition}

Next, we introduce useful lemmas and theorems. First, we show that the
type-directed transformation $Γ\proves c\transform c'$ is \emph{faithful}. In
other words, for any initial memory $m$ and program $c$, memory $m'$ is a
possible final memory iff for initial extended memory $\extmemfull{m}{0}{0}$
and $c'$, one final memory is an extension of $m'$.
\begin{lemma}[Faithfulness]
\begin{multline*}
∀m,m',c,c',Γ.~Γ\proves c\transform c' ⇒ \\
\evalexpr{c}{m}(m')\not= 0 ⇔ ∃\priv.~\extmemfull{m'}{\priv}{\delta}\in \evalexpr{c'}{\extmemfull{m}{0}{0}}
\end{multline*}
\end{lemma}
\begin{proof}
By structural induction on $c$.
\end{proof}

For a pair of initial and final memories $m_0$ and $m'$ when executing
the original program, we identify a set of possible $\vpriv$ values, so that in the
corresponding executions of $c'$, the initial and final memories are extensions
of $m$ and $m'$ respectively:
\begin{definition} Given a target program $c'$, an initial memory $m_0$ and a
final memory $m'$, the consistent costs of executing $c'$ w.r.t.  $m_0$ and
$m'$, written $\restrict{c'}{m_0}{m'}$, is defined as follows
\[\restrict{c'}{m_0}{m'} \defn \{\priv' \mid \extmemfull{m}{\priv'}{\delta}\in
\evalexpr{c'}{\extmemfull{m_0}{0}{0}} ∧ m=m'\}\]
where $m=m'$ iff $∀x\in\dom{(m)}.~m'(x)=m(x)$
\end{definition}

Since $(\restrict{c'}{m_1}{m})$ by definition is a set of values of $\vpriv$,
we write $\max(\restrict{c'}{m_1}{m})$ for the maximum cost. The next lemma
enables precise reasoning of privacy cost w.r.t. a pair of initial and final
memories when $Γ$ is injective:

\begin{lemma}[Point-Wise soundness]
\label{lem:pointwise}
\begin{multline*}
∀c,c',m_1,m_2,m,Γ.~\proves Γ ∧ Γ\proves c \transform c' ∧ \store_1~Γ~\store_2,\text{ we have } \\
\evalexpr{c}{\store_1}(m)≤\exp(\max(\restrict{c'}{m_1}{m}))\evalexpr{c}{\store_2}(Γ(m))
\end{multline*}
\end{lemma}

The full proof of Lemma~\ref{lem:pointwise} is available in 
\ifreport
the appendix.
\else
the full version of this paper~\cite{report}. 
\fi
We comment that this point-wise result enables precise reasoning of privacy
cost where the composition theorem falls short. Consider the transformed Sparse
Vector method in Figure~\ref{fig:translated}. This point-wise result allows
various cost bounds to be provided for various memories: $\vpriv$ increments by
$2\priv/N$ when the branch condition is true, but it remains the same otherwise.  On the
other hand, methods based on the composition theorem
(e.g.,~\cite{Fuzz,DFuzz,Barthe12,Barthe14}) have to (conservatively) provide an
unique cost bound for all possible executions, rendering a cost of $2\priv/N$. 

The point-wise soundness lemma provides a precise privacy bound per initial and
final memory. However, differential privacy by definition
(Definition~\ref{def:diffpriv}) bounds the worst-case cost. To close the gap,
we define the worst-case cost of the transformed program.

\begin{definition}
For any program $c$ in the target language, we say $c$'s execution cost is
\emph{bounded} by some constants $\priv$, written
$c^{\preceq \priv}$, iff for \emph{any} $\extmemfull{m}{0}{0}$,
\[\extmemfull{m'}{\priv'}{\delta'}\in \evalexpr{c}{\extmemfull{m}{0}{0}} ⇒ \priv'≤\priv\]
\end{definition}

Note that this safety property can be verified by an external mechanism such as
Hoare logic and model checking. Off-the-shelf tools can be used to verify that
$c^{\preceq \priv}$ holds for some $\priv$. For example, we have formally
proved that the transformed program in Figure~\ref{fig:sparsevector} satisfies
a postcondition $\vpriv≤\priv$ by providing one line of annotation (the grey
line in Figure~\ref{fig:sparsevector}) using the Dafny tool~\cite{Dafny}. 

\begin{theorem}[Soundness]
\label{thm:soundness}
\begin{multline*}
∀c,c',m_1,m_2,Γ,\priv.~\proves Γ ∧ Γ\proves c \transform c' ∧ \store_1~Γ~\store_2∧c'^{\preceq \priv},\text{ we have } \\
\gammadis(\evalexpr{c}{m_1},\evalexpr{c}{m_2})≤0
\end{multline*}
\end{theorem}
\begin{proof}
By definition, $(\max (\restrict{c}{m_1}{m}))≤\priv$ for all $m$,
$m_1$. Hence by Lemma~\ref{lem:pointwise}, $∀m.~
\evalexpr{c}{\store_1}(m)≤\exp(\priv)\evalexpr{c}{\store_2}(Γ(m))$. Hence,
\begin{align*}
   &\max_{S⊆\Store}(\evalexpr{c}{m_1}(S)-\exp(\priv)\evalexpr{c}{m_2}(Γ(S))) \\
 = &\max_{S⊆\Store}\sum_{m\in S} (\evalexpr{c}{m_1}(m)-\exp(\priv)\evalexpr{c}{m_2}(Γ(m)))
 ≤ 0
\end{align*}
We note that the equality in the proof above holds due to the injective
assumption $(\proves Γ)$, which allows us to derive the set-based privacy from
the point-wise privacy (Lemma~\ref{lem:pointwise}).
\end{proof}

We now connect the soundness theorem to differential privacy:

\begin{theorem}[Privacy]
\label{thm:privacy}
\begin{multline*}
∀Γ,c,c',x,\priv.~\proves Γ ∧ Γ\proves (c;\outcmd~e) \transform (c';\outcmd~e) \text{ then }\\
 c'^{\preceq \priv}
 ⇒ c \text{ is } \priv\text{-differentially private}
\end{multline*}
\end{theorem}
\begin{proof}
Proof is available in 
\ifreport
the appendix.
\else
the full version of this paper~\cite{report}. 
\fi
\end{proof}

\section{Differential-Privacy Proof Inference}
\label{sec:fulllang}

We have so far presented an explicitly typed language \langzero. However,
writing down types (especially those dependent types) for variables is still a
non-trivial task. Moreover, when multiple proofs exist, writing down types
accompanied with the minimum privacy cost is even more challenging. We extend
\langzero to automatically infer a proof and even search for the optimal one.

\subsection{Type Inference}
\label{sec:challenges}

\begin{figure}
$
\begin{array}{l@{\ }l@{\ }c@{\ }l}
\text{Distance Vars} & α,β,γ \in \DVars \\
\text{Distances} &  \dexpr &::=\; \dots \mid α
\end{array}
$
\caption{\lang: language syntax extension.}
\label{fig:fullsyntax}
\end{figure}

Since each type has two orthogonal components (base type and distance),
inference is needed for both. The former is mostly standard (e.g., for
Hindley/Milner system~\cite{Wand-typeinference,aiken-typeinclusion,zm14}),
hence omitted in this paper. 

Next, we assume all base types are available, and focus on the inference of the
distance counterpart. For brevity, we write $Γ(x)=\dexpr$ instead of
$∃\basety.~Γ(x)=\basety_\dexpr$. We use $\cod{DefVars}$ to represent the set of
variables whose distances are given by the programmer. 

To enable type inference, we extend \langzero with distance variables such as
$α,β,γ$ (shown in Figure~\ref{fig:fullsyntax}). Initially, the typing
environment associates each variable in $\cod{DefVars}$ with its annotated
distance. It associates each other variable with a distinguished distance
variable to be inferred.

Following the idea of modeling type inference as constraint solving
(e.g.,~\cite{Wand-typeinference,aiken-typeinclusion,haack:slicing}), it is
straightforward to interpret the typing rules in Figure~\ref{fig:typing} as a
(naive) inference algorithm. To see how, consider two assignments $(x:=0;
y:=x)$, where $Γ(x)\!=\!α,Γ(y)\!=\!β$. With distance variables, the typing
rules now collect constraints (instead of checking their validity) during type
checking. For example, two constraints are collected for those two assignments:
$α=0$ and $β=α$. Hence, inferring types is equivalent to finding
a solution for those two constraints (i.e., the satisfiability problem of
$∃α,β.~α=0∧β=α$). It is easy to check that $α=0∧β=0$ is a solution. Hence, the
inferred distances are $Γ(x)=0,Γ(y)=0$. However, this naive inference algorithm
falls short in face of dependent types. Next, we first explore the main
challenges in inferring dependent types, and then propose our inference
algorithm.

\paragraph{Inferring star types}
Consider the example in Figure~\ref{fig:partialsum}. If we follow the naive
inference algorithm above, two constraints are generated from lines 2 and 4:
$α\!=\!0$ and $α=α+\distance{q}[i]$, where $α=Γ(\cod{sum})$. These constraints are
unsatisfiable, since the value of $\distance{q}[i]$ is an arbitrary value
between $-1$ and $1$.
Nevertheless, the powerful type system of \langzero still allows formal
verification of this example by assigning $\cod{sum}$ to the star type, meaning
that its distance is dynamically tracked.
 
We observe that starting from the initial typing environment, we can refine it
by processing each assignment $x := e$ in the following way. We first
\emph{synthesize} the type of $e$ from its subexpressions, in the same fashion
as the original typing rules in Figure~\ref{fig:typing}. Then, if $x\in
\cod{DefVars}$ (i.e., given by the programmer), there is nothing to be refined.
Otherwise, we can \emph{refine} the typing environment by updating the type of
$x$ to a more precise one:
\[
\refine(Γ,x,\dexpr)\defn \begin{cases}
                       \subst{Γ}{α}{\dexpr} &\mbox{if } Γ(x) = α \in \DVars \\
                       Γ &\mbox{if } Γ(x) \not\in\DVars ∧ (Γ(x)=\dexpr)\\
                       Γ[x\mapsto *] &\mbox{otherwise} \\
                       \end{cases}
\]
Here, the auxiliary function $\refine$ takes an initial environment $Γ$, a
variable $x$ and a distance expression $\dexpr$. 
This function replaces all occurrences of $α$ in $Γ$ to $\dexpr$ when $Γ(x)$ is a
variable to be inferred ($α\in \DVars$). Otherwise, it statically
checks whether the old and
new distance expressions are equivalent. When the equivalence cannot be
determined at static time, it assigns the $*$ type to $x$.

Our inference algorithm refines the typing environment as it proceeds.
Consider Figure~\ref{fig:partialsum} again. At line 4, $\cod{sum}$'s distance
is refined to 0. Then at line 6, its distance is refined to $*$, since we
cannot statically check that $0=0+\distance{q}[i]$ is valid.

\paragraph{Inferring dependency on program state} 
Consider Figure~\ref{fig:sparsevector} where only the type of $η_2$ is to be
inferred.  The naive inference algorithm will generate one constraint for the
branch condition in line 6:
\vspace{-2ex}
\begin{multline*}
∀q_i,\distance{q_i},η_2,\tT.~(-1≤\distance{q_i}≤1)⇒\\
(q_i+η_2≥\tT ⇔ (q_i+η_2+\distance{q_i}+α≥\tT+1)
\end{multline*}
which is unsatisfiable, since there is no single value $α$ that can hide the difference of
$q_i$ in both directions. We need a more precise type for $η_2$ (as provided in
Figure~\ref{fig:sparsevector}) so that the ``if'' and ``else'' branches can be
aligned in different ways.

To infer dependent types, our inference algorithm propagates context
information to subexpressions. In
particular, we observe that only rule~\ruleref{T-ODot} generates a constraint
that may benefit from dependency on program states. Hence, our inference
algorithm propagates the comparison result to its subexpressions, and refine
subexpressions (e.g., $η_2$) for the needed dependency.

\paragraph{Inference algorithm}
We now present our inference algorithm, which is still based on the typing
rules in Figure~\ref{fig:typing}.  However, to tackle the challenges above, we
run a \emph{refinement algorithm} before type inference. The algorithm is shown
in Figure~\ref{fig:refinement}.

\begin{figure}
\framebox{Refinement rules for expressions}
\begin{mathpar}
\inferrule*[right=(R-Num)]{ }{ Γ;\predicate  \infers \real : Γ}
\and
\inferrule*[right=(R-Boolean)]{ b\in \{\true, \false\}}{ Γ;\predicate  \infers b : Γ}
\and
\inferrule*[right=(R-Var)]{ }
              { Γ; \predicate  \infers x : Γ}
\and
\inferrule*[right=(R-Rand)]{\predicate= ∅}
                 { Γ;\predicate \infers η : Γ}
\and
\inferrule*[right=(R-Rand-Refine)]{\predicate \not= ∅ \\ α_t,α_f \text{ fresh variables}}
                 { Γ;\predicate \infers η : \refine(Γ, η, \predicate?α_t:α_f)}
\and
\inferrule*[right=(R-Ops)]{Γ; \predicate \infers e_1: Γ_1 \quad 
                 Γ_1; \predicate \infers e_2 : Γ_2 \quad \op\in \oplus∪\otimes}
                       { Γ; \predicate \infers e_1~\op~e_2 : Γ_2}
\and
\inferrule*[right=(R-ODot)]{\inferrule*{}{
            Γ;   \predicate ∧ (e_1\odot e_2) \infers e_1\!:\! Γ_1 \quad
            Γ_1; \predicate ∧ (e_1\odot e_2) \infers e_2\!:\! Γ_2}}
                             { Γ; \predicate \infers e_1\odot e_2 : Γ_2}
\\
\and
\inferrule*[right=(R-Cons)]{Γ;\predicate \infers e_1: Γ_1 \quad Γ_1;\predicate \infers e_2:Γ_2}
               {Γ;\predicate \infers e_1::e_2 : Γ_2}
\and
\inferrule*[right=(R-Neg)]{Γ;\predicate \infers e: Γ'}{ Γ;\predicate \infers \neg e\!:\! Γ'}
\quad
\inferrule*[right=(R-Idx)]{Γ;\predicate \infers e_1\!:\!Γ_1\quad Γ_1;\predicate\infers e_2\!:\!Γ_2}{Γ;\predicate \infers e_1[e_2]:Γ_2}
\end{mathpar}
\framebox{Refinement rules for commands}
\begin{mathpar}
\inferrule*[right=(R-Skip)]{ }{Γ \infers \skipcmd : Γ}
\and
\inferrule*[right=(R-Return)]{ }
                             {Γ \infers \outcmd{e}: Γ}
\and
\inferrule*[right=(R-Asgn-Ref)]{ x\not\in \cod{DefVars} \and Γ,∅ \infers e : Γ' \and Γ' \proves e: \dexpr }
                           {Γ \infers x := e :\refine(Γ',x,\dexpr)}
\and
\inferrule*[right=(R-Asgn)]{ x\in \cod{DefVars} \quad Γ,∅ \infers e : Γ'}
                           {Γ \infers x := e : Γ'}
\and
\inferrule*[right=(R-Seq)]{Γ \infers c_1:Γ' \quad Γ' \infers c_2:Γ''}{Γ \infers c_1;c_2:Γ''}
\\
\inferrule*[right=(R-If)]{Γ;∅\infers e : Γ_1 \quad Γ_1 \infers c_1:Γ_2 \quad Γ_2 \infers c_2:Γ_3}
{ Γ \infers \ifcmd{e}{c_1}{c_2}:Γ_3}
\\
\inferrule*[right=(R-While)]{Γ;∅\infers e:Γ_1 \and Γ_1≤Γ_2 \and Γ_2 \infers c:Γ_2}{Γ \infers \whilecmd{e}{c} :Γ_2}
\end{mathpar}
\framebox{Refinement for random assignments}
\begin{mathpar}
\inferrule*[right=(R-Laplace)]{ }
                           {Γ \infers η := \lapm~r : Γ}
\end{mathpar}
\caption{The refinement algorithm.}
\label{fig:refinement}
\end{figure}

For expressions, the refinement algorithm propagates context information
$\predicate$ to subexpressions. Hence, each rule for expression has the form of
$Γ,\predicate \infers e : Γ'$, where $\predicate$ is a predicate that may appear in 
a dependent type, $Γ$ is the typing environment to be refined, and $Γ'$ is the
refined environment. The context information $\predicate$ is used to refine
distance of a random variable $η$ in rule~\ruleref{R-Rand-Refine}. Note that
the refinement is not needed for a normal variable $x$ (rule~\ruleref{R-Var}).
Intuitively, the reason is that the ``shape'' of $x$ is either provided or has
been refined when $x$ is initialized. However, this is not true for a random
variable: $η$ can have any distance expression according to
rule~\ruleref{T-Laplace}.

The refinement rules for commands have the form of $Γ\infers c : Γ'$. As we
described informally above, rule~\ruleref{R-Asgn-Ref} refines the distance of
$x$ using the $\refine$ function when its distance is not given.
The rule~\ruleref{R-While} assumes that a fixed point exists. Based on the
definition of the $\refine$ function, a fixed point can be computed
as follows. We define $≤$ as the lifted relation based on a point-wise lattice
(for each variable) where: $∀α,β\in\DVars.~α≤β∧β≤α$ and $α≤\dexpr≤*$ if $\dexpr$ is not
a distance variable. We can compute a fixed point by $Γ_1=Γ^0\proves c:Γ^1,
Γ^1\proves c:Γ^2,\cdots$ until $Γ^i\proves c:Γ^i$ for some $i$. Based on the
definition of the $\refine$ function, it is easy to check that $Γ^i≤Γ^{i+1}$
and the computation terminates since whenever $Γ^i\not=Γ^{i+1}$, either the
number of distance variables is reduced by one, or one more variable has a star
type.

\paragraph{Example} We consider type inference for our running example in
Figure~\ref{fig:sparsevector} where all local variables are to be inferred. We
first run the refinement algorithm. The first refinement happens at line 2,
where the distance of $\tT$ is refined to $α$, the distance variable of $η_1$.
At line 3, $\cod{c_1}$, $\cod{c_2}$ and $\cod{i}$ are refined to distance $0$.
In the loop body, $η_2$ is refined to $q[i]+η_2≥\tT?β:γ$ at line 6, using rule~\ruleref{R-Rand-Refine}. At line 8,
$\refine(Γ,\cod{c1},0)$ returns $Γ$ since $0=0$ is always true. Similar for the
``else'' branch and line 12. Hence, the environment after line 12 is already a fixed
point for the loop body. Hence, the typing environment after refinement is:
$Γ(\cod{c1})=Γ(\cod{c2})=Γ(i)=0$,  $Γ(\tT)=Γ(η_1)=α$ and
$Γ(η_2)=(q[i]+η_2≥\tT)?β:γ$.

\paragraph{Type checking with distance variables}
With type variables in the refined environment $Γ$, the type system collects
constraints during type checking, and tries to solve the collected constraints
where the type variables are existentially qualified. For example, with type
refinement, type checking the partial sum example in
Figure~\ref{fig:partialsum} yields a unique solution, which is identical to the
type annotation in the figure. In general, collected constraint may have
multiple solutions. For example, type checking the Sparse Vector method
generates only one (nontrivial) constraint from the rule~\ruleref{T-ODot}:
\vspace{-2ex}
\begin{multline*}
\big(∀q_i,\distance{q_i},η_2,\tT\in \mathbb{R}.~(-1≤\distance{q_i}≤1)⇒ \\
q_i+η_2≥\tT ⇔
q_i+\distance{q_i}+η_2+(q_i+η_2≥\tT?β:γ) ≥ \tT+α\big)
\end{multline*}

It is easy to check that the type annotation in Figure~\ref{fig:sparsevector}
(i.e., $α=1,β=2,γ=0$) is a solution of the constraint.  But in fact, other
solutions exist. For example, $α=0,β=2,γ=-2$ and $α=2,β=3,γ=0$ are both valid
solutions. The type system can either pick a solution, or defer the inference
by transforming the original program to a target program where type variables
are treated as unknown program inputs (as shown in
Figure~\ref{fig:translatedfinal}).

\subsection{Minimizing Privacy Cost}
\begin{figure}
\small
\algrule
\funsigfour{MSparseVec}{$T,N:\tyreal$; $q:\tylist~\tyreal$; \instrument{\distance{q}: \tylist~\tyreal};\\ $α,β,γ:\tyreal$}
       {$\cod{out} : \tylist~\tyreal$}{$∀i.~‐1≤(\distance{q}[i])≤1$}
\algrule
\begin{lstlisting}[frame=none]
$\instrument{\vpriv := 0;}$
$\instrument{\havoc{η_1};\vpriv := \vpriv+(α\priv/2);}$
$\instrument{\tT := T + η_1;}$
c1 := 0; c2 := 0; i := 0;
while (c1 $<$ N)
  $\annotation{\cod{Invariant}: \cod{c1}≤N ∧ \vpriv = \frac{α\priv}{2}+\cod{c1}\times \frac{β\priv}{2N}+\cod{c2}\times \frac{γ\priv}{2N}}$
  $\instrument{\havoc{η_2};\vpriv := \vpriv + (q[i]+η_2≥\tT?β:γ)\times \priv/4c;}$
  if ($q[i]+η_2≥\tT$) then
    out:= true::out;
    c1 := c1+1;
  else
    out:= false::out;
    c2 := c2+1;
  i := i+1;
\end{lstlisting}
\caption{The target program with unknown type variables. The instrumented statements are underlined.}
\label{fig:translatedfinal}
\end{figure}

With type variables captured explicitly in the transformed program, we can
verify that the postcondition $\vpriv =
\frac{α\priv}{2}+\frac{β\priv}{2}+\cod{c2}\times \frac{γ\priv}{2N}$ holds by
providing the loop invariant shown in grey.  Hence, combined with the
remaining unsolved constraints on those type variables, finding the optimal
proof is equivalent to the following MaxSMT problem, where $M$ is a large
number since $\cod{c2}$ is not bounded:
\begin{multline*}
\min (\frac{α}{2}+\frac{β}{2}+M\times γ) \text{ such that } \\
\big(∀q_i,\distance{q_i},η_2,\tT\in \mathbb{R}.~(-1≤\distance{q_i}≤1)⇒ \\
   q_i+η_2≥\tT ⇔ q_i+\distance{q_i}+η_2+(q_i+η_2≥\tT?β:γ) ≥ \tT+α\big)
\end{multline*}
Using a MaxSMT solver $μZ$~\cite{vZ-SCSS,vZ-TACAS}, we successfully find the
optimal solution for the type variables: $α=1,β=2,γ=0$. This is exactly the
randomness alignment used in its formal proof~\cite{diffpbook}. 

We note that the translation to the MaxSMT problem at this stage still requires
programmer efforts (e.g., identifying the cost bound involving type variables and
converting the cost bound to an equivalent formula suitable for a MaxSMT
solver).  However, this example clearly demonstrates the potential benefits of
explicitly calculating the privacy cost in the target language. 

\subsection{Proof Automation}

In general, a \lang-based proof consists four steps involving manual efforts:
\begin{inparaenum}[1)]
\item writing down the program specification (i.e., the function signature that
specifies private and non-private parameters and return values),
\item writing down the type annotations for local variables, 
\item verifying that the privacy cost in the transformed program is bounded by
either a known budget, or (MaxSMT only) a formula involving unsolved type
variables, and
\item (MaxSMT only) solving the MaxSMT problem of ``$\min$(upper bound formula)
such that constraints from step 2 are satisfiable''. 
\end{inparaenum}

As most verification tools, \lang requires a programmer to write down
specification (step 1). For step 2, we find that the inference algorithm in
Section~\ref{sec:challenges} is powerful enough to automatically infer the
types for the nontrivial algorithms considered in this paper\footnote{The only
exception is the algorithm in Section~\ref{sec:categorical}, since the
algorithm uses a uniformly distributed random source which is currently absent
in the inference algorithm.}. For step 3 and step 4, \lang relies on the
automation in existing verification tools.
We note that though \lang currently adds no automation in step 3 and 4,
separating relational reasoning from counting privacy cost and automating task
2 greatly simplifies those steps for all examples that we have seen so far. We
leave systematic research in automating the entire proof as future work.

\section{Case Studies}\label{sec:casestudy}
\input{casestudy}

\section{Related Work}\label{sec:related}
\input{related}

\section{Conclusions and Future Work}\label{sec:conclusions}
\input{conclusions}

\section*{Acknowledgments}

We thank Adam Smith, our shepherd Marco Gaboardi and anonymous reviewers for
their helpful suggestions. This work was supported by NSF grants CNS-1228669
and CCF-1566411.

\balance
\bibliographystyle{abbrvnat}
\bibliography{diffpriv}

\ifreport
\clearpage
\appendix
\nobalance
\addcontentsline{toc}{section}{Appendices}
\ifreport
\section*{Appendix}
\else
\section*{Supplementary material of submission \#234}
\fi


\section{Soundness Proof}

In the source language semantics (Figure~\ref{fig:semantics}), the variable
that a star-typed variable depends on (e.g., $\distance{\cod{sum}}$ in
Figure~\ref{fig:partialsum}) is invisible.  We first extend the semantics to
make the manipulation of such invisible variables explicit, by the following
rule for assignments:
\begin{align*}
\trans{x:=e}_\store &= \unitop~(\subst{\subst{m}{x}{\evalexpr{e}{m}}}{\distance{x}}{\evalexpr{\dexpr}{m}}) \text{ when } Γ(x)=\tyreal_*,
\end{align*}
where $Γ(e)=\dexpr$. 

It is straightforward to check that the extended semantics (parameterized on
the type system) is consistent with the original semantics in
Figure~\ref{fig:semantics}, as it does not change the distribution on the
variables that are visible in the source program. The extended semantics is
needed to close the gap between the source language and the one that formal
reasoning is applied on.

Next, we prove a few auxiliary lemmas.

\begin{lemma}
\label{lem:point}
\begin{multline*}
∀m_1,m_2,Γ~s.t.~m_1~Γ~m_2, \text{ we have } \\
∀m.~\unitop~{m_1}(m)=\unitop~{m_2}(Γ(m))
\end{multline*}
\end{lemma}
\begin{proof}
By the fact that $Γ$ is a function.
\end{proof}

\begin{lemma}[Expression]
\label{lem:expr}
\begin{multline*}
∀e,\store_1,\store_2, Γ~s.t.~\store_1~Γ~\store_2 ∧ Γ \proves e : \basety_{\dexpr}, \text{ we have } \\
\evalexpr{e}{\store_1} + \evalexpr{\dexpr}{\store_1} = \evalexpr{e}{\store_2}
\end{multline*}
\end{lemma}
\begin{proof}
Induction on the structure of $e$. Interesting cases are follows.

When $e$ is $x$ or $η$, result is true by the definition of $m_1~Γ~m_2$. 

When $e$ is $e_1\oplus e_2$, let $\evalexpr{e_i}{\store_1}=v_i$,
$\evalexpr{e_i}{\store_2}=v_i'$ and $Γ\proves e_i:\basety^i_{\dexpr_i}$ for
$i\in\{1,2\}$. Then by typing rule, we have $\dexpr=\dexpr_1+\dexpr_2$.
By induction hypothesis, we have $v_i'=v_i+d_i$, where
$d_i=\evalexpr{\dexpr_i}{\store_1}$. Hence, $\evalexpr{e}{m_2}=v_1'\oplus
v_2'=(v_1+d_1)\oplus(v_2+d_2)=(v_1\oplus v_2)+(d_1\oplus
d_2)=\evalexpr{e_1\oplus e_2}{m_1} +
\evalexpr{\dexpr_1\oplus\dexpr_2}{m_1}=\evalexpr{e}{m_1}+\evalexpr{\dexpr}{m_1}$.

When $e=e_1\odot e_2$, let $Γ\proves e_i:\basety^i_{\dexpr_i}$ for
$i\in\{1,2\}$. Then by induction hypothesis, we have
$\evalexpr{e_i}{m_1}+\evalexpr{\dexpr_i}{m_1}=\evalexpr{e_i}{m_2}$ for
$i\in\{1,2\}$. By rule~\ruleref{T-ODot}, for any memory $m$, $\evalexpr{e_1\odot
e_2}{m}=\evalexpr{(e_1+\dexpr_1)\odot (e_2+\dexpr_2)}{m}$.  Hence,
$\evalexpr{e_1\odot e_2}{m_1}=\evalexpr{(e_1+\dexpr_1)\odot
(e_2+\dexpr_2)}{m_1}=\evalexpr{e_1+\dexpr_1}{m_1}\odot
\evalexpr{e_2+\dexpr_2}{m_1}=\evalexpr{e_1}{m_2}\odot \evalexpr{e_2}{m_2}=
\evalexpr{e_1\odot e_2}{m_2}$.  
\end{proof}

\paragraph{Proof of Lemma~\ref{lem:pointwise}}
\begin{multline*}
∀c,c',m_1,m_2,m,Γ.~\proves Γ∧ Γ\proves c \transform c' ∧ \store_1~Γ~\store_2,\text{ we have } \\
\evalexpr{c}{\store_1}(m)≤\exp(\max (\restrict{c'}{m_1}{m}))\evalexpr{c}{\store_2}(Γ(m))
\end{multline*}
\begin{proof} By structural induction on $c$.
\begin{itemize}
\item
Case $\skipcmd$: $c'=\skipcmd$ by typing rule. Hence, $\max
(\restrict{c'}{m_1}{m})=0$. Desired result is true by Lemma~\ref{lem:point} and
the semantics of $\skipcmd$.

\item
Case $x:=e$: by the transformation, we have  $\max (\restrict{c'}{m_1}{m})=0$.
Hence, by the semantics and Lemma~\ref{lem:point}, it is sufficient to show
that the memories after the assignment are related by $Γ$.

We first show $m_1'(x)+\evalexpr{\dexpr}{m_1'} =m_2'(x)$.

\begin{itemize}
\item When $Γ(x)=\basety_{\dexpr}$, we need to show that $m_1'~Γ~m_2'$ where
$m_1'=\subst{m_1}{x}{\evalexpr{e}{m_1}}$ and
$m_2'=\subst{m_2}{x}{\evalexpr{e}{m_2}}$ by the semantics.
By typing rule, we have $Γ\proves e:\basety_{\dexpr}$ as well. By
Lemma~\ref{lem:expr}, $\evalexpr{e}{m_1}+\evalexpr{\dexpr}{m_1} =
\evalexpr{e}{m_2}$. Hence, we have $m_1'(x)+\evalexpr{\dexpr}{m_1}=m_2'(x)$.
Since $\dexpr$ may only depend on immutable variables in this case,
$\evalexpr{\dexpr}{m_1}=\evalexpr{\dexpr}{m_1'}$. So
$m_1'(x)+\evalexpr{\dexpr}{m_1'}=m_2'(x)$ as desired. 

\item When $Γ(x)=\basety_*$, $Γ\proves x:\basety_{\distance{x}}$. Hence,
$m_1'(x)+m_1'(\distance{x})=\evalexpr{e}{m_1}+\evalexpr{\dexpr}{m_1}$, where
$Γ(e)=\dexpr$, by the extended semantics. By Lemma~\ref{lem:expr}, this is
identical to $\evalexpr{e}{m_2}$, which is $m_2'(x)$ by the semantics.
\end{itemize}

Second, we show $m_1'(y)+\evalexpr{\dexpr}{m_1'} =m_2'(y)$, where $Γ\proves
y:\dexpr$ for $y\in\dom(Γ) ∧ y\not=x$. When $y\in\Vars$, its type cannot depend
on $x$, which is mutable. So the desired result is true. For $η\in Η$, its type
only depends on the memory state when $η$ is used. So the desired result is
true as well. 

\item
Case $\ifcmd{e}{c_1}{c_2}$: by typing rule, $Γ \proves e:\bool_0$. By
Lemma~\ref{lem:expr}, $\evalexpr{e}{m_1}=\evalexpr{e}{m_2}$. Hence, the same
branch is taken in $m_1$ and $m_2$. Desired result is true by induction
hypothesis.

\item
Case $c_1;c_2$: For any $m$ such that $\evalexpr{c_1;c_2}{m_1}(m)\not= 0$,
there exists some $m'$ such that 
\[\evalexpr{c_1}{m_1}(m')\not=0 ∧ \evalexpr{c_2}{m'}(m)\not=0\] 
By induction hypothesis, we have
\[\evalexpr{c_1}{m_1}(m')≤\exp(\priv_1)\evalexpr{c_1}{Γ(m_1)}(Γ(m'))\]
\[\evalexpr{c_2}{m'}(m)≤\exp(\priv_2)\evalexpr{c_2}{Γ(m')}(Γ(m))\]
where $\priv_1=\max (\restrict{c_1'}{m_1}{m'})$ and 
$\priv_2=\max (\restrict{c_2}{m'}{m})$. Hence,
\begin{multline*}
\evalexpr{c_1}{m_1}(m')\cdot\evalexpr{c_2}{m'}(m)≤ \\
\exp(\priv_1+\priv_2)\evalexpr{c_1}{Γ(m_1)}(Γ(m'))\cdot\evalexpr{c_2}{Γ(m')}(Γ(m))
\end{multline*}
Notice that $\extmemfull{m'}{\priv_1}{0}\in
\evalexpr{c_1}{\extmemfull{m_1}{0}{0}}$ and $\extmemfull{m}{\priv_2}{0}\in
\evalexpr{c_2}{\extmemfull{m'}{0}{0}}$ since $\priv_1$ and $\priv_2$ maximize
privacy costs among consistent executions by definition. Hence,
$\extmemfull{m}{\priv_1+\priv_2}{0}\in
\evalexpr{c_1';c_2'}{\extmemfull{m_1}{0}{0}}$. Therefore,
$\priv_1+\priv_2≤\max (\restrict{c_1;c_2}{m_1}{m})$.

So for any $m$,
\begin{align*}
\evalexpr{c_1;c_2}{m_1}(m) &= \sum_{m'} \evalexpr{c_1}{m_1}(m')\cdot \evalexpr{c_2}{m'}(m) \\
&≤ \exp(\priv') \sum_{m'} \evalexpr{c_1}{Γ(m_1)}(Γ(m'))\cdot \evalexpr{c_2}{Γ(m')}(Γ(m))\\
&≤ \exp(\priv') \sum_{m'} \evalexpr{c_1}{Γ(m_1)}(m')\cdot \evalexpr{c_2}{m'}(Γ(m))\\
&≤ \exp(\priv') \evalexpr{c_1;c_2}{m_2}(Γ(m))
\end{align*}
where $\priv'=\max (\restrict{c_1;c_2}{m_1}{m})$. Notice that the change
of variable in the second to last inequality only holds when $Γ$ is an
injective (but not necessarily onto) mapping, which is true due to the
assumption $\proves Γ$.

\item
Case $\whilecmd{e}{c}$: let $W=\whilecmd{e}{c}$. By typing rule, $Γ\proves
e:\bool_0$. Hence, $\evalexpr{e}{m_1'}=\evalexpr{e}{m_2'}$ for any
$m_1'~Γ~m_2'$. We proceed by by natural induction on the number of loop
iterations (denoted by $i$) under $m_1$.

When $i=0$, $\evalexpr{b}{m_1}=\false$. So $\evalexpr{b}{m_2}=\false$ since
$m_1~Γ~m_2$. By semantics, $\evalexpr{W}{m_1}=\unitop~m_1$ and
$\evalexpr{W}{m_2}=\unitop~m_2$, and $\max (\restrict{W}{m_1}{m})=0$.
Desired result is true by Lemma~\ref{lem:point}.

Consider $i=j+1$. $\evalexpr{b}{m_1}=\true$. So $\evalexpr{b}{m_2}=\true$ since
$m_1~Γ~m_2$.  By semantics, $\evalexpr{W}{m_i}=\evalexpr{c;W}{m_i}$ for $i\in
\{1,2\}$, and the latter $W$ iterates for $j$ times. By induction hypothesis
and a similar argument as the sequential case,
$\evalexpr{W}{m_1}(m)≤\exp(\max
(\restrict{W}{m_1}{m})\evalexpr{W}{m_2}(Γ(m))$.

\item
Case $η:=\lapm r$: let $μ_r=\lapm r$. Since $μ_r$ is the Laplace distribution with
a scale factor of $r$, we have 
\[∀v,d\in \mathbb{R}.~μ_r(v) ≤\exp(|d|\times r) μ_r (v+d)\]

When ∄$v.~m=\subst{m_1}{η}{v}$ , $\evalexpr{η:=\lapm r}{m_1}(m)=0$ by the
semantics. Hence, desired inequality is trivial. 

When $m=\subst{m_1}{η}{\cod{v}}$ for some constant $\cod{v}$, we have for any
$d\in \mathbb{R}$,
\begin{align*}
  &\evalexpr{η:=\lapm r}{m_1}(m) \\
=~&μ_r (\cod{v})\\
≤~&\exp(|d|\cdot r) μ_r(\cod{v}+d)
\end{align*}

Let $Γ(η)=\tyreal_{\dexpr}$ and $\evalexpr{\dexpr}{m}=\cod{d}$ for some
constant $\cod{d}$.  Since $m_1~Γ~m_2$,
$\subst{m_1}{η}{\cod{v}}~Γ~\subst{m_2}{η}{\cod{v+d}}$. That is,
$Γ(m)=\subst{m_2}{η}{\cod{v+d}}$.  By the semantics, 
\begin{align*}
\evalexpr{η:=\lapm r}{m_2}(Γ(m)) &=\evalexpr{η:=\lapm r}{m_2}(\subst{m_2}{η}{\cod{v+d}}) \\
                                 &=μ_r(\cod{v+d})
\end{align*}
Hence, we have
\[\evalexpr{η:=\lapm r}{m_1}(m)≤\exp(|\cod{d}|\cdot r) \evalexpr{η:=\lapm
r}{m_2}(Γ(m))\]
when $m=\subst{m_1}{η}{\cod{v}}$ for some constant $\cod{v}$ too.
By the typing rule~\ruleref{T-Laplace}, the transformed program is
$(\havoc{η};\vpriv := \vpriv+|\dexpr|\cdot r)$. Hence, $\max
(\restrict{c_1'}{m_0}{m})=\evalexpr{|\dexpr|\cdot r}{m} =
|\evalexpr{\dexpr}{m}| \cdot r=|\cod{d}|\cdot r$.
Therefore, we showed that
\begin{multline*}
\evalexpr{η:=\lapm r}{m_1}(m)≤ \\
\exp(\max (\restrict{c_1'}{m_0}{m}))\evalexpr{η:=\lapm r}{m_2}(Γ(m))
\end{multline*}
\end{itemize}
\end{proof}

\paragraph{Proof of Theorem~\ref{thm:privacy}}
\begin{multline*}
∀Γ,c,c',x,\priv.~\proves Γ ∧ Γ\proves (c;\outcmd~e) \transform (c';\outcmd~e) \text{ then }\\
 c'^{\preceq \priv}
 ⇒ c \text{ is } \priv\text{-private}
\end{multline*}
\begin{proof}
By the soundness theorem (Theorem~\ref{thm:soundness}), we have for any injective $Γ$,
$m_1~Γ~m_2$, $∀S⊆\Store, \evalexpr{c}{m_1}(S) ≤ \exp(\priv)
\evalexpr{c}{m_2}(Γ(S))$. For clarity, we stress that all sets are over
distinct elements (as we have assumed throughout this paper).
Let $P=(c;\outcmd~e)$. By typing rule~\ruleref{T-Return}, the return type
$\basety$ must be either $\tyreal$ or $\bool$, and its distance must be zero.
By semantics, for any value set $V⊆ \basety$, 
\begin{align}
\evalexpr{P}{m_1}(V) &=  \evalexpr{c}{m_1}(\{m \mid \evalexpr{e}{m} \in V\}) \\
                     &≤  \exp(\priv) \evalexpr{c}{m_2}(\{Γ(m) \mid \evalexpr{e}{m}\in V\}) \\
                     &≤  \exp(\priv) \evalexpr{c}{m_2}(\{m \mid \evalexpr{e}{m}\in V\}) \\
                     &=  \exp(\priv) \evalexpr{P}{m_2}(V)
\end{align}
where inequality (2) is true due to Theorem~\ref{thm:soundness} (the
application of which requires the injective assumption). For inequality
(3), consider any $m'\in \{Γ(m) \mid \evalexpr{e}{m} \in V\}$. It must be
true that $m'=Γ(m)∧\evalexpr{e}{m} \in V$ for some $m\in \Store$. Due to
Lemma~\ref{lem:expr}, $\evalexpr{e}{m}=\evalexpr{e}{Γ(m)}$ (the distance of
$e$ must be 0). That is, $\evalexpr{e}{m} =v ⇔ \evalexpr{e}{Γ(m)} =v$ for any $v$. Hence, 
\[m'=Γ(m)∧\evalexpr{e}{m}\in V \text{ for some } m\in \Store\]
is the same as 
\[m'=m''∧\evalexpr{e}{m''}\in V  \text{ for some } m\in \Store, \text{ where } m''=Γ(m)\]
Since $m''\in \Store$, $m'\in \{m \mid \evalexpr{e}{m}\in V\}$. Hence, the inequality (3)
holds. We note that (3) is not an equality in general since $Γ$ might not be a
surjection.

Therefore, by definition of differential privacy, $c$ is $\priv$-private.
\end{proof}

\section{Formal Semantics for the Target Language}

The denotational semantics interprets a command $c$ in the target language
(Figure~\ref{fig:targetlang}) as a function $\trans{c}:
\Store→\mathcal{P}(\Store)$.  The semantics of commands are formalized as
follows. 
\begin{align*}
\trans{\skipcmd}_\store &= \{m\} \\
\trans{x:=e}_\store &= \{\subst{m}{x}{\evalexpr{e}{m}}\} \\
\trans{\havoc{x}}_\store &= ∪_{r\in \mathbb{R}}~\{\subst{m}{x}{r}\} \\
\trans{c_1;c_2}_\store &= ∪_{m'\in \evalexpr{c_1}{\store}} \evalexpr{c_2}{m'}  \\
\trans{\ifcmd{e}{c_1}{c_2}}_\store &= 
      \begin{cases} 
      \trans{c_1}_\store &\mbox{if } \evalexpr{e}{\store} = \true \\
      \trans{c_2}_\store &\mbox{if } \evalexpr{e}{\store} = \false \\
      \end{cases} \\
\trans{\whilecmd{e}{c}}_\store &= w^*~m\\
\text{where } w^*              &= fix (λf.~λm. \cod{if}~\evalexpr{e}{m}=\true\\
                               &\qquad \cod{then}~(∪_{m'\in \evalexpr{c}{m}}~f~m')~\cod{else}~{\{m\}}) \\
\trans{c;\outcmd{e}}_\store &= ∪_{m'\in \evalexpr{c}{\store}}~\{\evalexpr{e}{m'}\}
\end{align*}

Accordingly, the Hoare logic rules for the target language is mostly standard,
summarized in Figure~\ref{fig:hoare}.

\section{Uniform Distribution}
\begin{lemma}[UniformDist]
The following typing rule is sound w.r.t. Lemma~\ref{lem:pointwise}:
\begin{mathpar}
\inferrule*[]{ Γ(η) = \tyreal_{η\cdot \dexpr}\and -1<\dexpr≤0}
                           {Γ \proves η := \uniform \transform \cod{havoc}_{[0,1]}{η}; \vpriv = \vpriv-\log (\dexpr+1)}
\end{mathpar}
\end{lemma}
\begin{proof}
When ∄$v.~m=\subst{m_1}{η}{v}$ or $m(η)<-1$ or $m(η)>1$,
$\evalexpr{η:=\uniform}{m_1}(m)=0$ by the semantics. Hence, desired inequality
is trivial. 

When $m=\subst{m_1}{η}{\cod{v}}$ for some $-1≤\cod{v}≤1$. Let $μ=\uniform$,
$\cod{d}=\evalexpr{\dexpr}{m}$. Notice that by typing rule $\dexpr≤0$. So
$\cod{d}≤0$. We have
\begin{align*}
  &\evalexpr{η:=\uniform}{m_1}(m) \\
=~&μ (\cod{v}) \\
=~&\int_0^1 \dgdist_{\{x≤\cod{v}\}} dx \\
=~&\int_0^1 \dgdist_{\{(\cod{d}+1)x≤\cod{(d+1)v}\}} dx \\
=~&\int_0^{1+\cod{d}} \frac{1}{1+\cod{d}}\dgdist_{\{x≤(\cod{d}+1)\cod{v}\}} dx \\
≤~&\frac{1}{1+\cod{d}} \int_0^1 \dgdist_{\{x≤(\cod{d}+1)\cod{v}\}} dx \\
=~&\exp(-\log (\cod{d}+1)) μ(\cod{v+v\cdot d})
\end{align*}

Since $m_1~Γ~m_2$, we have $\subst{m_1}{η}{\cod{v}}~Γ~\subst{m_2}{η}{(\cod{v+v\cdot d})}$.Hence,
\begin{multline*}
\evalexpr{η:=\uniform}{m_1}(m)≤ \\
\exp(-\log(\cod{d}+1))\evalexpr{η:=\uniform}{m_2}(Γ(m))
\end{multline*}

By the transformation, $\uniform \transform \cod{havoc}_{[0,1]}{η}; \vpriv =
\vpriv-\log (\dexpr+1)$, where $Γ(η)=η\cdot\dexpr$. Hence, $\max
(\restrict{c_1'}{m_0}{m})=\evalexpr{-\log(\dexpr+1)}{\subst{m_1}{η}{\cod{v}}}=-\log(\cod{d}+1)$.
Therefore, 
\begin{multline*}
\evalexpr{η:=\uniform}{m_1}(m)≤\\
\exp(\max (\restrict{c_1'}{m_0}{m}))\evalexpr{η:=\uniform}{m_2}(Γ(m))
\end{multline*}

\end{proof}

\begin{figure}
\begin{mathpar}
\inferrule*[right=(H-Skip)]{ }{\htriple{\post}{\skipcmd}{\post}}
\and
\inferrule*[right=(H-Asgn)]{ }{\htriple{\subst{\post}{x}{e}}{x:=e}{\post}}
\and
\inferrule*[right=(Havoc)]{ }{ \htriple{∀x.~\post}{\havoc(x)}{\post}}
\and
\inferrule*[right=(H-Seq)]{\htriple{\pre}{c_1}{\post'} \quad \htriple{\post'}{c_2}{\post}}{\htriple{\pre}{c_1;c_2}{\post}}
\and
\inferrule*[right=(H-Return)]{ }{\htriple{\post}{\outcmd{e}}{\post}}
\and
\inferrule*[right=(H-If)]{\htriple{e∧\pre}{c_1}{\post} \\ \htriple{\neg e∧\pre}{c_2}{\post}}
{ \htriple{\pre}{\ifcmd{e}{c_1}{c_2}}{\post}}
\and
\inferrule*[right=(H-While)]{\pre⇒\inv \\ \htriple{\inv∧e}{c}{\inv} \\ \inv⇒\post}
                            {\htriple{\pre}{\whilecmd{e}{c}}{\post}}
\end{mathpar}
\caption{Hoare logic rules for the target language.}
\label{fig:hoare}
\end{figure}

\fi
\end{document}

%% file: abstract.tex
The growing popularity and adoption of differential privacy in academic and
industrial settings has resulted in the development of increasingly
sophisticated algorithms for releasing information while preserving privacy.
Accompanying this phenomenon is the natural rise in the development and
publication of incorrect algorithms, thus demonstrating the necessity of formal
verification tools. However, existing formal methods for differential privacy
face a dilemma: methods based on customized logics can verify sophisticated
algorithms but come with a steep learning curve and significant annotation
burden on the programmers, while existing programming platforms lack expressive
power for some sophisticated algorithms.

In this paper, we present \lang, a simple imperative language that strikes a
better balance between expressive power and usability. The core of \lang is a
novel relational type system that separates relational reasoning from privacy
budget calculations. With dependent types, the type system is powerful enough
to verify sophisticated algorithms where the composition theorem falls short.
In addition, the inference engine of \lang infers most of the proof details,
and even searches for the proof with minimal privacy cost bound when multiple
proofs exist. We show that \lang verifies sophisticated algorithms with little
manual effort.

%% file: intro.tex
Companies, government agencies, and academics are interested in analyzing and modeling datasets containing sensitive information about individuals (e.g., medical records, customer behavior, etc.). Privacy concerns can often be mitigated if the algorithms used to manipulate the data, answer queries, and build statistical models satisfy differential privacy \cite{dwork06Calibrating} --- a set of restrictions on their probabilistic behavior that provably limit the ability of attackers to infer individual-level sensitive information \cite{dwork06Calibrating,pufferfishjournal}.

Since 2006, differential privacy has seen explosive growth in many areas, including theoretical computer science, databases, machine learning, and statistics. This technology has been deployed in practice, starting with the U.S. Census Bureau LEHD OnTheMap tool \cite{ashwin08:map}, the Google Chrome Browser \cite{rappor}, and Apple's new data collection efforts \cite{appledp}. 
However, the increase in popularity and usage of differential privacy has also been accompanied by a corresponding increase in the development and implementation of algorithms with flawed proofs of privacy; for example, \citet{ashwinsparse} and \citet{ninghuisparse} catalog some recent cases about variations of the Sparse Vector method~\cite{diffpbook} .

Currently, there are two strategies for combating this trend. The first is the use of programming platforms \cite{pinq,gupt,Roy10Airavat} that have privacy primitives that restrict the privacy-preserving algorithms that can be implemented and often add more noise than is necessary to the computation. The second strategy is the development of languages and formal verification tools for differential privacy \cite{Fuzz,DFuzz,Barthe12,Barthe14,Barthe16,BartheCCS16}. These languages enable the development of much more sophisticated algorithms that use less noise and hence provide more accurate outputs. However, the increased power of the formal methods comes with a considerable cost --- a programmer has to heavily annotate code and generate proofs using complicated logics such as a customized relational Hoare logic proposed by~\citet{Barthe12}. Moreover, intricate proof details have to be provided by a programmer, which makes exploring variations of an algorithm difficult since small variations in code can cause significant changes to a proof.

In this paper, we present \lang, a language for developing provably privacy-preserving algorithms. The goal of \lang is to minimize the burden on the programmer while retaining most of the capabilities of the state-of-the-art, such as verifying the Sparse Vector method \cite{diffpbook} (an algorithm which, until very recently~\cite{Barthe16,BartheCCS16}, was beyond the capabilities of verification tools). For example, we show that the Sparse Vector method can be verified in \lang with little manual effort: just two lines of annotation from the programmer.

\lang is equipped with a novel light-weight relational type system that clearly
separates relational reasoning from privacy budget calculation. In particular,
it transforms the original probabilistic program into an equivalent
nonprobabilistic program, where all privacy costs become explicit. With
dependent types, the explicitly calculated privacy cost in the target language
may depend on program states, hence enabling the verification of sophisticated
algorithms (e.g., the Sparse Vector method) that are beyond the capability of
many existing methods~\cite{Fuzz,DFuzz,Barthe12,Barthe14}  based on the
composition theorem~\cite{pinq}.  Moreover,  the transformed nonprobabilistic
program is ready for off-the-shelf formal verification methods, such as Hoare
logic, to provide an upper bound of the privacy cost.  

On the usability end, \lang has an inference engine that reduces the already
low annotation burden on the programmers. Although the inference engine
does not yet automate the privacy budget calculation part of a proof, it does
 fill in missing details in the relational reasoning part of a
proof; furthermore, based on MaxSMT theory, it even searches for the optimal proof that
minimizes privacy cost with minimal human involvement. For example, with only one
postcondition annotation and one loop invariant annotation from a programmer,
\lang confirms that the proof in~\cite{diffpbook} indeed provides the minimal
privacy cost.

To summarize, this paper makes the following contributions:
\begin{enumerate}
\item \lang, a new imperative language for verifying sophisticated
privacy-preserving algorithms (Section~\ref{sec:syntax}), 

\item expressive static annotations incorporating dependent types, enabling
precise tracking of privacy costs (Section~\ref{sec:typing}),

\item a formal proof that the \lang type system soundly tracks differential
privacy costs, and new proof techniques involved
in the soundness proof (Section~\ref{sec:soundness}),

\item an inference engine that automatically fills in missing details involved
in the relational reasoning part of a proof, and further, minimizes provable
privacy cost bound when multiple proofs exist, with little manual effort
(Section~\ref{sec:fulllang}),

\item case studies on complex algorithms showing that formal verification of
privacy-preserving algorithms are viable with little programmer annotation
burden (Section~\ref{sec:casestudy}). 
\end{enumerate}

%% file: preliminaries.tex
\subsection{Distributions}
\label{sec:dist}
We define the set of \emph{sub-distributions} over a discrete set $A$, written
$\dist(A)$, as the set of functions $μ: A → [0,1]$, such that $\sum_{a\in A}
μ~a≤1$. When applied to an event $E⊆A$, we define $μ(E)\defn \sum_{e\in E}
μ(e)$. Notice that we do not require $\sum_{a\in A} μ~a=1$, a special case when
$μ$ is a \emph{distribution}, since sub-distribution gives rise to an elegant
semantics for programs that may not terminate~\cite{Kozen81}.

Given a distribution $μ\in \dist(A)$, its support is defined as
$\support(μ)\defn \{a\mid μ(a)>0\}$.  We use $\dgdist_a$ to represent the
degenerate distribution $μ$ that $μ(a)=1$ and $μ(a')=0$ if $a'\not=a$.
Moreover, sub-distributions can be given a structure of a monad. Formally, we
define the $\unitop$ and $\bindop$ functions as follows:
\begin{align*}
\unitop &: A → \dist(A) \defn λa.~\dgdist_a \\ 
\bindop &: \dist(A)→(A → \dist(B))→\dist(B)  \\
        &\defn λμ.~λf.~(λb.~\sum_{a\in A} (f~a~b)\times μ(a))
\end{align*}
That is, $\unitop$ takes an element in $A$ and returns the Dirac distribution
where all mass is assigned to $a$; $\bindop$ takes $μ$, a distribution on
$A$, and $f$, a mapping from $A$ to distributions on $B$ (e.g., a conditional
distribution of $B$ given $A$), and returns the
corresponding marginal distribution on $B$. This monadic view will avoid
cluttered definitions and proofs when probabilistic programs are involved
(Section~\ref{sec:semantics}). 

\subsection{Differential Privacy}

Differential privacy has two major variants: \emph{pure} \cite{dwork06Calibrating} (obtained by setting $\delta=0$ in the following definition) and \emph{approximate} \cite{dworkKMM06:ourdata} (obtained by choosing a $\delta>0$).

\begin{definition}[Differential privacy]
\label{def:diffpriv}
Let $\priv,\delta≥0$. A probabilistic computation $M: A \rightarrow \dist(B)$ is
$(\priv,\delta)$-differentially private with respect to an adjacency relation
$\pre\subseteq A\times A$ if for every pair of inputs $a_1,a_2\in A$ such that $a_1 \pre
a_2$, and every output subset $E\subseteq B$, we have
\[P(M a_1\in E)\leq\exp(\priv)P(M a_2\in E)+\delta\]
\end{definition}

Intuitively, a probabilistic computation satisfies differential privacy if it
produces similar distributions for any pair of inputs related by $\pre$.  In
the most common applications of differential privacy, $A$ is the set of
possible databases and the adjacency relation $\pre$ is chosen so that $a_1\pre
a_2$ whenever $a_1$ can be obtained from $a_2$ by adding or removing data
belonging to a single individual.

In this paper, we focus on the verification of algorithms that satisfy pure
differential privacy. An algorithm is $\priv$-differentially private iff it is
$(\priv,0)$-differentially private according to Definition~\ref{def:diffpriv}.

%% file: illustrative.tex
\subsection{The Sparse Vector Method}

\newcommand{\algrule}[1][.5pt]{\par\vskip.2\baselineskip\hrule height #1\par\vskip.2\baselineskip}

\begin{figure}
\small

\algrule
\funsigfour{SparseVector}{$T,N,\priv$\annotation{:\tyreal_0}; $q$\annotation{:\tylist~\tyreal_*}}{(out\annotation{:\tylist~\bool_0})}{$∀i.~‐1≤(\distance{q}[i])≤1$}
\algrule
$\annotation{\cod{c1,c2,i}:\tyreal_0; \tT, η_1:\tyreal_1;η_2:\tyreal_{q[i]+η_2≥\tT?2:0}}$
\begin{lstlisting}[frame=none]
$η_1$ := $\lapm{(2/\priv)}$;
$\tT$ := $T + η_1;$
c1 := 0; c2 := 0; i := 0;
while (c1 < $N$)
  $η_2$ := $\lapm{(4N/\priv)}$;
  if ($q[i]+η_2≥\tT$) then
    out:= true::out;
    c1 := c1 + 1;
  else
    out:= false::out;
    c2 := c2 + 1;
  i := i+1;
\end{lstlisting}
\caption{The Sparse Vector method. Type annotations are shown in grey. The
precondition specifies the adjacency assumption.}
\label{fig:sparsevector}
\end{figure}

The goal of formal methods for verifying $\epsilon$-differential privacy is to provide an upper bound on the privacy cost $\epsilon$ of a program. Typically, users will have a fixed privacy budget $\epsilon^\prime$ and can only run programs whose provable privacy cost $\epsilon$ does not exceed the budget: $\epsilon\leq \epsilon^\prime$. For this reason, it is important that formal methods are able to prove a tight upper bound on the privacy cost.

With the exception of \cite{Barthe16,BartheCCS16}, most existing formal methods rely on the composition theorem \cite{pinq}. That is, in the case of $\epsilon$-differential privacy, those methods essentially treat a program as a series of modules, each with a provable upper bound $\epsilon_i$ on its privacy cost. Then by the composition theorem~\cite{pinq}, the total privacy cost is bounded by $\sum_i \epsilon_i$. However, for sophisticated advanced algorithms, the composition theorem often falls short --- it can provide upper bounds that are arbitrarily larger than the true privacy cost. Providing the tightest privacy cost for intricate algorithms requires formal methods that are more powerful but avoid over-burdening the programmers with annotation requirements.
To illustrate these challenges, we consider
the Sparse Vector method \cite{diffpbook}. It is a prime example of the need for formal methods because many of its published variants have been shown to be incorrect \cite{ashwinsparse,ninghuisparse}.

The Sparse Vector method also has many correct variants, one of which is
shown in Figure~\ref{fig:sparsevector}. For now, safely ignore the type
annotations in grey and the precondition. Here, the input list $q$ represents a sequence of results of counting
queries $q_1, q_2, q_3, \dots$ (e.g., how many patients in the data have
cancer, how many patients contracted an infection in the hospital, etc.)
running on a database. The goal is to answer as accurately as possible the
following question: which queries, when evaluated on the true database, return an
answer greater than the threshold $T$ (a program input unrelated to the sensitive data)? 

To achieve differential privacy, the algorithm adds appropriate Laplace noise
to the threshold and to each query. Here, $\lapm(4N/\priv)$ draws one sample
from the Laplace distribution with mean zero and a scale factor $(4N/\priv)$.
If the noisy query answer ($q[i]+\eta_2$) is above the noisy threshold
$\tT$, it adds $\true$ to the output list $\cod{out}$ (in the slot reserved for that query) and
otherwise, it adds $\false$. The key to this algorithm is the deep
observation that once noise has been added to the threshold, queries for which
we output $\true$ have a privacy cost (so we can answer at most $N$ of them,
where $N$ is a parameter); however, outputting $\false$ for a query does not
introduce any new privacy costs \cite{diffpbook}. The algorithm ensures that
the total privacy cost is bounded by the input $\priv$, the parameter used in
Figure \ref{fig:sparsevector}. This remarkable property makes the Sparse Vector
method ideal in situations where the vast majority of query counts are expected
to be below the threshold.

\paragraph{Failure of the composition theorem}
If we just use the composition theorem, we would have a privacy cost of
$\epsilon/4N$ for each loop iteration (i.e., every time $\lapm(4N/\epsilon)$
noise is added to a query answer) due to the property of the Laplace
distribution. Since the number of iterations are not a priori bounded, the
composition theorem could not prove that the algorithm
satisfies $\epsilon^\prime$-differential privacy for any finite
$\epsilon^\prime$; more advanced methods are needed.

\paragraph{Informal proof and sample runthrough}
Proofs of correctness (of this and other variants) can be found in
\cite{diffpbook,ashwinsparse,ninghuisparse}. Here we provide an informal correctness argument by example to illustrate the subtleties involved both in proving it and inferring a tight bound for the algorithm.

Suppose we set the parameters $T=4$ (we want to know which queries have a value at least $4$) and $N=1$ (we stop the algorithm after the first time it outputs $\true$). Consider the following two databases $D_1, D_2$ that differ on one record, and their corresponding query answers:
\begin{align*}
D_1:\quad q[0]=2,\quad q[1]=3,\quad q[2]=5\\
D_2:\quad q[0]=3,\quad q[1]=3,\quad q[2]=4
\end{align*}
Suppose in one execution on $D_1$, the noise added to $T$ is $\alpha^{(1)}=1$ and the noise added to $q[0],q[1],q[2]$ is $\beta^{(1)}_0=2,\beta^{(1)}_1=0,\beta^{(1)}_2=0$, respectively. Thus the noisy threshold is $\tT=5$ and the noisy query answers are $q[0]+\beta^{(1)}_0=4$, $q[1]+\beta^{(1)}_1=3$, $q[2]+\beta^{(1)}_2=5$ and so the algorithm outputs the sequence: $(\false, \false, \true)$.

According to Definition~\ref{def:diffpriv}, for any output sequence $\omega$, we need to show $P(M(D_1)=\omega)\leq e^\epsilon P(M(D_2)=\omega)$ for all possible outputs $\omega$ and databases $D_1,D_2$ that differ on one record. For the databases $D_1,D_2$ described above, we will show that $P(M(D_1)=(\false,\false,\true))\leq e^\epsilon P(M(D_2)=(\false,\false,\true))$. We proceed in two steps.

\paragraph{Aligning randomness}
We first create an \emph{injective} (but not necessarily bijective) function from the randomness in the execution under $D_1$ into the randomness in the execution under $D_2$, so that both executions generate the same output. For an execution under $D_2$, let $\alpha^{(2)}$ be the noise added to the threshold and let $\beta^{(2)}_0,\beta^{(2)}_1,\beta^{(2)}_2$ be the noise added to the queries $q[0],q[1],q[2]$, respectively. Consider an injective function candidate  that adds 1 to the threshold noise (i.e., $\alpha^{(2)} = \alpha^{(1)}+1$), keeps the noise of queries for which $D_1$ reported $\false$ (i.e., $\beta^{(2)}_0 =\beta^{(1)}_0$ and $\beta^{(2)}_1 = \beta^{(1)}_1$) and adds 2 to the noise of queries for which $D_1$ reported $\true$ (i.e., $\beta^{(2)}_2 = \beta^{(1)}_2 +2$).

In our running example, execution under $D_2$ with this function would result in the noisy threshold $\tT=6$ and noise query answers  $q[0]+\rightv{β}_0=5$, $q[1]+\rightv{β}_1=3$, $q[2]+\rightv{β}_2=6$. Hence, the output once again is $(\false,\false,\true)$. In fact, it is easy to see that under this injective function, every execution under $D_1$ would result in an execution under $D_2$ that produces the same answer.

\paragraph{Counting privacy cost}
For each output $ω$, let $f$ be the injective function; let $A$ be the set of random variable assignments that cause execution under $D_1$ to produce $\omega$; let $B$ be the possible assignments we can get by applying $f$ to $A$; and let $C$ be the set of random variable assignments that are not in the range of $f$, but nevertheless cause execution under $D_2$ to produce the output $\omega$ as well\footnote{This is possible since we do not assume the function to be a bijection.}. Then we can rewrite $P(M(D_1)=\omega)=P(A)$ and $P(M(D_2)=\omega)=P(B)+P(C)$.

Once we have done the alignment of randomness as above, and recalling that $N=1$ in our example, the proof finishes by showing:
\begin{align*}
P(A) &= \sum_{a\in A} P(a) \leq e^{\frac{\epsilon}{2}}e^{2 \frac{\epsilon}{4N}} \sum_{a\in A} P(f(a))\\
     &= e^\epsilon P(B) \leq e^{\epsilon}(P(B)+P(C))
\end{align*}
where the $e^{\frac{\epsilon}{2}}$ factor results from using a threshold value that is $1$ larger, while the $e^{2 \frac{\epsilon}{4N}}$ factor results from adding 2 to the noise for query $q[2]$.
 Notice that no privacy cost is paid for queries $q[0]$ and $q[1]$, since the same noise is added under $D_1$ and $D_2$. Moreover, due to the injective assumption, $\sum_{a\in A} P(f(a))=P(B)$.

\paragraph{Challenges} 
The Sparse Vector method is a prime example of the need for formal methods,
since paper-and-pencil proof is shown to be error-prone for its variants. The
intricacy in its proof brings major challenges for formal methods:
\begin{enumerate}
\item Precision: a crucial observation in the proof is that once noise has been
added to the threshold, different privacy costs are only paid for outputting
$\true$. Hence, the cost calculation needs to consider program states.

\item Aligning randomness: finding an injective function from the randomness
under $D_1$ to that under $D_2$ such that outputting $ω$ under $D_1$ entails
outputting $ω$ under $D_2$ is the most intriguing piece in the proof. However,
coming up with a correct function, as we described informally above, is
non-trivial.

\item Finding the tightest bound: in fact, an infinite number of proofs exist
for the Sparse Vector method, though with various provable privacy
costs\footnote{For example, another injective function adds 2 to the threshold
noise (i.e., $\alpha^{(2)} = \alpha^{(1)}+2$), keeps the noise of queries for
which $D_1$ reported $\false$ (i.e., $\beta^{(2)}_0 =\beta^{(1)}_0$ and
$\beta^{(2)}_1 = \beta^{(1)}_1$) and adds 3 to the noise of queries for which
$D_1$ reported $\true$ (i.e., $\beta^{(2)}_2 = \beta^{(1)}_2 +3$). It is easy
to check that this mapping has the desired property, but the privacy cost is
$(7\priv)/4$.}. Since a tighter privacy cost bound allows a privacy-preserving
algorithm to produce more accurate outputs, a formal method should produce the
tightest one when possible.
\end{enumerate}

Except the very recent work by~\citet{Barthe16,BartheCCS16}, existing formal
methods (e.g.,~\cite{Fuzz,DFuzz,Barthe12,Barthe14}) rely on the composition
theorem, hence fail to prove that the Sparse Vector method satisfies
$\priv'$-privacy for any $\priv'$. Recent work by~\citet{Barthe16,BartheCCS16}
verifies variants of the Sparse Vector method, but its customized relational
logic incurs heavy annotation burden, including the randomness alignment.
Moreover, those work cannot search for the tightest cost bound.

\subsection{Our Approach}
To tackle the challenges above, we propose \lang, an imperative language that
enables verification and even inference of the tightest privacy cost for
sophisticated privacy-preserving algorithms. We illustrate the key components
of \lang in this section, and detail all components in the rest of this paper.

\paragraph{Relational reasoning}

The core of \lang is a novel light-weight dependent type system that explicitly
captures the \emph{exact difference} of a variable's values in two executions
under two adjacent databases. Let $v^{(1)}$ ($v^{(2)}$) be the value of a
variable $x$ in an execution under $D_1$ ($D_2$). The type $τ$ for $x$ in \lang
has the form of $\basety_{\dexpr}$, meaning that $x$ holds a value of basic
type $\basety$ (e.g., $\inte$, $\cod{real}$, $\bool$), and
$v^{(1)}+\dexpr=v^{(2)}$. Required type annotations for the Sparse Vector
method are shown in grey in Figure~\ref{fig:sparsevector}. Note that for
brevity, we write $\tyreal$ for numeric base types (e.g., $\inte$,
$\cod{real}$).  Hereafter, we refer to the $\dexpr$ counterpart as
the \emph{distance}.

In the simplest case, the distance is a constant. For example, the input
threshold $T$ has a type $\tyreal_0$, meaning that its value remains the same
in two executions (since $T$ is a parameter unrelated to private information
about individuals). The distance of variable $η_1$ captures one randomness
alignment in the informal proof above: we enforce a distance of $1$ for $η_1$
(i.e., we map noise $v$ to $v+1$ for any value $v$ sampled in the execution
under $D_1$). 

The distance may also depend on program states. Hence, \lang supports dependent
types. For instance, consider the distance of $η_2$: $q[i]+η_2≥\tT?2:0$. This
annotation specifies an injective function that maps noise $v$ to $v+2$ when the
value of $q[i]+η_2≥\tT$ is $\true$ (i.e., the output is $\true$) in an
execution under $D_1$, and maps noise $v$ to $v$ otherwise. As we will see
shortly, dependent types allow a precise privacy cost to be calculated
under various program states, hence enabling bounding privacy cost in a
tighter way than mechanisms based on the composition theorem.

Moreover, \lang uses a distinguished distance $*$ as a shorthand for the
standard Sigma type (e.g., $\tyreal_* \defn \Sigma_{x:\tyreal_0}~\tyreal_x$).
In other words, each value of type $\tyreal_*$ (e.g., each query in the list
$q$) can be interpreted as a pair of the form $(x:\tyreal_0,y:\tyreal_x)$,
where the first component specifies the distance of the second component. Note
that by language design, the first component is invisible in the
privacy-preserving algorithm; however, the type system may reason about and
manipulate it via a distinguished operation~$\distance{ }$ .
For instance, the precondition in the running example states the adjacent
assumption on databases: for each query answer $q[i]$ in $q$, its distance
($\distance{q}[i]$) is bounded by $±1$. The star type is also useful for
distances that cannot be easily captured at compile time
(Section~\ref{sec:syntax}).

With type annotations, a type system statically verifies that the distances are
maintained as an invariant throughout the execution. For example, the output
$\cod{out}$ in Figure~\ref{fig:sparsevector} has type $\tylist~\bool_0$,
meaning that each element in the list has type $\bool_0$. Hence, the
invariant maintained by the type system ensures that  two related
executions always generate the same output. 

\paragraph{Calculating privacy cost} When type-checking succeeds, the type
system transforms the original program to a non-probabilistic, non-relational
program where privacy cost is explicitly calculated. The transformed program
for the Sparse Vector method is shown in Figure~\ref{fig:translated}. 

The transformed program is almost identical to the original one, except that:
\begin{inparaenum}[1)]
\item the privacy cost is explicitly calculated via a new variable $\vpriv$,
and
\item probabilistic instructions are replaced by a new nondeterministic
instruction $\havoc{η}$, which semantically sets variable $η$ to an arbitrary
value upon execution.
\end{inparaenum}

\begin{figure}
\small
\algrule
\funsigfour{TSparseVec}{$T,N,\priv:\tyreal$; $q:\tylist~\tyreal$; \instrument{\distance{q}:\tylist~\tyreal}}{($\cod{out}:\tylist~\bool$)}{$∀i.~‐1≤(\distance{q}[i])≤1$}
\algrule
\begin{lstlisting}[frame=none]
$\instrument{\vpriv := 0;}$
$\instrument{\havoc{η_1};\vpriv := \vpriv+\priv/2;}$
$\tT := T + η_1;$
c1 := 0; c2 := 0; i := 0;
while (c1 $<$ N)$\Suppressnumber$
  $\annotation{\cod{Invariant}: \cod{c1}≤N ∧ \vpriv = \priv/2+\cod{c1}\times \frac{\priv}{2N}}$$\Reactivatenumber$
  $\instrument{\havoc{η_2};\vpriv := \vpriv + (q[i]+η_2≥\tT?2:0)\times \priv/4N;}$
  if ($q[i]+η_2≥\tT$) then
    out:= true::out;
    c1 := c1+1;
  else
    out:= false::out;
    c2 := c2+1;
  i := i+1;
\end{lstlisting}
\caption{The transformed program. The instrumented statements are underlined.
The loop invariant to prove the postcondition
$\protect\vpriv\protect≤\protect\priv$ is shown in grey.}
\label{fig:translated}
\end{figure}

The fundamental soundness theorem of the type system states that, informally,
if
\begin{inparaenum}[1)]
\item the original program type-checks, and
\item $\vpriv$ is always bounded by some constant $\priv'$ in the transformed
program,
\end{inparaenum}
then the original program being verified is $\priv'$-differentially private.

Notice that for the second property (the problem of bounding $\vpriv$), any
off-the-shelf verification tool for functional correctness can be utilized.  For
instance, the program above with the desired postcondition $\vpriv≤\priv$ can
be verified by Hoare logic, with one loop invariant provided by a
programmer (the grey box in Figure~\ref{fig:translated}).

\subsection{Type Inference}

The mechanisms sketched so far provide a light-weight yet powerful formal
method for differential privacy. The annotation burden is much reduced compared
with~\cite{Barthe16}. However, providing the \emph{correct} and \emph{optimal}
type annotations (especially for random variables $η_1$ and $η_2$) is still
subtle for a programmer. 

Although it is folklore that type inference in face of dependent types can be daunting, \lang is equipped
with an inference engine that, at least for many algorithms, not only infers correct annotations, but also
enables finding annotations that minimize privacy cost when multiple annotations exist.  For example,
given the function signature in Figure~\ref{fig:sparsevector}, the inference
algorithm in Section~\ref{sec:fulllang} automatically infers types for all
variables. The inferred types are identical to the ones in
Figure~\ref{fig:sparsevector} except that $\tT$, $η_1$ are assigned with a
distance $α$ and $η_2$ is assigned with a distance expression
$q[i]+η_2≥\tT?β:γ$, where $α,β,γ$ are variables to be inferred, subject to
constraints generated during type checking. For the Sparse Vector method,
multiple correct type annotations exist. For example, $α=1,β=2,γ=0$ corresponds
to the annotation in Figure~\ref{fig:sparsevector}. Moreover, $α = 0,β = 2,γ = −2$
and $α = 2,β = 3,γ = 0$ are both correct annotations.

Given type annotations with distance variables, the type-guided transformation as
sketched above generates target program where variables to be inferred (i.e.,
$α,β,γ$) are used in the calculation of privacy cost. The difference is that
$\vpriv$ increments by $(α\priv)/2$ at line 2 and increments by
$(q[i]+η_2≥\tT?β:γ)\times \priv/4N$ at line 6 in Figure~\ref{fig:translated}.
By Hoare logic, we can easily bound the privacy cost to be $α\priv/2 + β
\priv/4 + \cod{c2}\times γ \priv/4N$.

Putting it all together, finding the optimal proof is equivalent to a MaxSMT problem:
$\min{(2α + β + \cod{c2}\times γ/N)}$, given that constraints generated in type
checking are satisfiable. Using an existing MaxSMT solver,
$μZ$~\cite{vZ-SCSS,vZ-TACAS}, the optimal proof for the Sparse Vector method is
successfully inferred: $α=1,β=2,γ=0$. This is exactly the randomness alignment
used in its proof~\cite{diffpbook}.

%% file: casestudy.tex
\subsection{Sparse Vector with Numerical Answers}
\label{sec:numericalsparse}

\begin{figure}
\small
\algrule
\funsigfour{NumSparseVector}{$T,N,\priv$\annotation{:\tyreal_0}; $q$\annotation{:\tylist~\tyreal_*}}{($\cod{out}:\annotation{\tylist~\tyreal_0}$)}{$∀i.~‐1≤(\distance{q}[i])≤1$}
\algrule
$\framebox{\ensuremath{\cod{c1,c2,i}:\tyreal_0;\tT, η_1:\tyreal_1;η_2:\tyreal_{q[i]+η_2≥\tT?2:0};η_3:\tyreal_{-\distance{q}[i]}}}$
\begin{lstlisting}[frame=none]
$η_1$ := $\lapm{(3/\priv)}$;
$\tT$ := $T + η_1;$
c1 := 0; c2 := 0; i := 0;
while (c1 < N)
  $η_2$:= $\lapm{(6N/\priv)}$;
  if ($q[i]+η_2≥\tT$) then 
    $η_3$:= $\lapm{(3N/\priv)}$;
    out:= ($q[i]$+$η_3$)::out;
    c1 := c1 + 1;
  else
    out:= 0::out;
    c2 := c2 + 1;
  i := i+1;
\end{lstlisting}
The transformed program, where underlined commands are added by the type
system. Only one annotation (loop invariant) is needed from the programmer to
verify the postcondition $\vpriv≤\priv$: $\cod{c1}≤N ∧ \vpriv =
\frac{\priv}{3}+\cod{c1}\times \frac{2\priv}{3N}$.

\begin{lstlisting}[frame=none]
$\instrument{\vpriv := 0;}$
$\instrument{\havoc{η_1}; \vpriv := \vpriv + \priv/3;}$
$\tT$ := $T+η_1$;
c1 := 0; c2 := 0; i := 0;
while (c1 < N)
  $\instrument{\havoc{η_2}; \vpriv := \vpriv + (q[i]+η_2≥\tT?2:0)\times \priv/6N;}$
  if $(q[i]+η_2≥\tT)$ then
    $\instrument{\havoc{η_3}; \vpriv := \vpriv + |\distance{q}[i]|\times \priv/3N;}$
    out:= ($q[i]$+$η_3$)::out;
    c1 := c1 + 1;
  else
    out:= 0::out;
    c2 := c2 + 1;
  i := i+1
\end{lstlisting}
\caption{The Numerical Sparse Vector method.}
\label{fig:sparsevector2}
\end{figure}

We first study a numerical variant of the Sparse Vector method. The previous
version (Figure \ref{fig:sparsevector}), produces only two types of outputs for
each query: $\true$, meaning the query answer is probably above the threshold;
and $\false$, meaning that it is probably below. The numerical variant, shown
in Figure~\ref{fig:sparsevector2}, replaces the output $\true$ with a noisy
query answer. It does this by drawing fresh Laplace noise and adding it to the
query (Line 8). 

\paragraph{Verification using \lang} \lang can easily verify this numerical
variant from scratch, in a very similar way as verifying the Sparse Vector method.
However, here we focus on another interesting scenario of using \lang: the
programmer (or algorithm designer) has already verified the Sparse Vector
method using \lang, and she is now exploring its variations. This is a common
scenario for algorithm designers. We show that since \lang automatically fills
in most proof details, exploring variations of an algorithm requires little
effort.

In particular, we assume the programmer has already obtained the (optimal)
types for all local variables except $η_3$, and the loop invariant shown in
Figure~\ref{fig:translated} from the verification of the Sparse Vector method.
Hence, the type inference engine only needs to infer a type for $η_3$, which is
trivially solved to be $\tyreal_{-\distance{q}[i]}$. Moreover, \lang transforms
the original program to the one on the bottom of
Figure~\ref{fig:sparsevector2}. To finish the proof, according to
Theorem~\ref{thm:privacy}, it is sufficient to verify the postcondition that
$\vpriv≤\priv$. In fact, only one annotation (shown in
Figure~\ref{fig:sparsevector2}) that is very close to the one in
Figure~\ref{fig:translated} is needed to finish the proof. Hence, we just
proved the numerical Sparse Vector variant for (almost) free using \lang.

\paragraph{Incorrect variants}
The numerical variant is also historically interesting since it fixes a bug in
a very influential set of lecture notes \cite{aaronnotes}; these lecture notes
inadverantly re-used the same noise used for the ``if'' test (Line 7) instead
of drawing new noise when outputting the noisy query answer. In other words,
Lines 5-8 in Figure~\ref{fig:sparsevector} are replaced with:
\begin{lstlisting}[frame=none,numbers=none]
  $η_2$:= $\lapm{(2N/\priv)}$;
  $\tilde q = q[i]+η_2$
  if ($\tilde q≥\tT$) then 
    out:= ($\tilde q$)::out;
\end{lstlisting}
For this incorrect variant, the refinement algorithm refines the type of
$\tilde q$ to be $\distance{q}[i]+α$ when $\tilde q$ is defined, where
$Γ(η_2)=α$.  Moreover, during type checking, ($(\tilde q)::\cod{out}$)
generates a constraint ($\distance{q}[i]+α=0$) by rule~\ruleref{T-Cons}. Hence,
it must be true that $Γ(\tilde q)=\tyreal_0$ and
$Γ(η_2)=\tyreal_{-\distance{q}[i]}$ after type inference.
Moreover, after type checking, $\cod{η_2:= \lapm{(2N/\priv)}}$ is transformed
to \[(\cod{havoc~η_2; \vpriv := \vpriv + |\distance{q}[i]|(\priv/2N)})\]
However, we cannot prove that the incorrect variant is $\priv'$-private for any
$\priv'$. The reason is that $\vpriv$ in the transformed program is clearly not
bounded by any constant $\priv'$: $\vpriv$ increments by $\priv/2N$ in the
worst case in each loop iteration, but the number of iterations is unbounded
(when most iterations take the ``else'' branch). 

The failure of a formal proof of the incorrect variant also sheds lights on how
to fix it. For example, if we bound the number of iterations to be $N$, then
the incorrect variant is fixed (though with a different privacy cost).

\subsection{Smart Summation}
\label{sec:smartsum}
We next study a smart summation algorithm verified previously (with heavy
annotations) in~\cite{Barthe12,Barthe14}. The pseudo code, shown in Figure
\ref{fig:smartsum}, is adapted from~\cite{Barthe14}. The goal of this smart sum
algorithm is to take a finite sequence of bits $q[0], q[1], \dots, q[T]$ and
output a noisy version of their partial sum sequence: $q[0], ~q[0] + q[1], ~
\dots, ~\sum_{i=0}^T q[i]$. One naive approach is to add Laplace noise to each
partial sum (partial implementation is shown in Figure~\ref{fig:partialsum}).  
An alternative naive algorithm is to compute a noisy bit $\tilde{q}[i] = q[i] +
\laplace(1/\priv)$ for each $i$ and output $\tilde{q}[0], ~\tilde{q}[0] +
\tilde{q}[1], ~ \dots, ~\sum_{i=0}^T \tilde{q}[i]$. However, in both
approaches, the noise will swamp the true counts.

A much smarter approach was proposed by \citet{chan11continualTISSEC}.
Intuitively, their algorithm groups $q$ into nonoverlapping blocks of size $M$.
So block $G_1=\{q[0], q[1],\dots, q[M-1]\}$, $G_2=\{q[M], q[M+1], q[M+2],
\dots, q[2M-1]\}$, etc.  Then it maintains 2 levels of noisy counts: (1) the
noisy bits $\tilde{q}[i] = q[i]+\laplace(1/\priv)$ for each $i$, and (2) the
noisy block sums $\tilde{G}_j = \sum_{i\in G_j}q[i]+\laplace(1/\priv)$ for each
block. The partial sums are computed from these noisy counts in the following
way. Consider the sum of the first $\ell+1$ bits: $\sum_{i=0}^\ell q[i]$. We
can represent $\ell+1 = xM + c$ where $x = \lfloor \frac{\ell+1}{M}\rfloor$ and
$c = \ell+1 \mod M$. Hence, the  noisy partial sum can be computed from the
noisy sum of the first $x$ blocks plus the remaining $c$ noisy bits: 
$\tilde{G}_1 + \tilde{G}_2 +\dots + \tilde{G}_x +
\sum\limits_{j=0}^{c-1}\tilde{q}[xB+j]$. This algorithm is shown in Figure
\ref{fig:smartsum}. 
The ``if'' branch keeps track of block boundaries and is responsible for
summing up the noisy blocks. The ``else'' branch is responsible
for adding in the remaining loose noisy bits (once there are enough loose bits
to form a new block $G_j$, we use its noisy sum $\tilde{G}_j$ rather than the
sum of its noisy bits).

\begin{figure}
\small
\algrule
\funsigfour{SmartSum}{$\priv$, M, T:$\annotation{\tyreal_0}$; q:\annotation{\tylist~\tyreal_*}}{($\cod{out}:\annotation{\tylist~\tyreal_0}$)}{$∀i.~‐1≤(\distance{q}[i])≤1 ∧\\ \makebox[2.5cm]{} (∀i.~\distance{q}[i]\not=0 ⇒ (∀j\not=i.~\distance{q}[j]=0))$}
\algrule
$\framebox{\ensuremath{\cod{next, n, i}:\tyreal_0;\cod{sum}:\tyreal_*;η_1:\tyreal_{-\distance{\cod{sum}}-\distance{q}[i]};η_2:\tyreal_{-\distance{q}[i]}}}$
\begin{lstlisting}[frame=none]
next:=0; n:=0; i:=0; sum := 0;
while i $\leq$ T
  if (i + 1) mod M = 0 then
    $η_1$ := $\lapm{1/\priv}$;
    n := n + sum + q[i] + $η_1$;
    next:= n;
    sum := 0;
    out := next::out; 
  else
    $η_2$ := $\lapm{1/\priv}$;
    next:= next + q[i] + $η_2$;
    sum := sum + q[i];
    out := next::out;
  i := i+1;
\end{lstlisting}
The transformed program, where underlined commands are added by the type
system. Only one annotation (loop invariant) is needed from the programmer to
verify the postcondition
$\vpriv≤2\priv$: $(\vpriv+\distance{\text{sum}}>0⇒∀j≥i.~\distance{q}[j]=0)∧
(\distance{\text{sum}}>0⇒\vpriv≤\priv) ∧ (\distance{\text{sum}}≤1.0) ∧
(\vpriv≤2\priv)$
\begin{lstlisting}[frame=none]
$\instrument{\vpriv = 0;}$
next:=0, n:=0, i:= 0, sum:=0;$\instrument{\distance{\text{sum}} := 0}$;
while i $\leq$ T
   if (i + 1) mod M = 0 then
      $\instrument{\havoc{η_1};\vpriv := \vpriv + |\distance{\cod{sum}}+ \distance{q}[i]|\priv;}$
      n := n + sum + q[i] + $η_1$;
      next:= n;
      sum := 0;
      $\instrument{\distance{\text{sum}} := 0;}$
      out := next::out; 
   else
      $\instrument{\havoc{η_2};\vpriv := \vpriv + |\distance{q}[i]|\priv};$
      next:= next + q[i] + $η_2$;
      sum := sum + q[i];
      $\instrument{\distance{\text{sum}} := \distance{\text{sum}}+|\distance{q}[i]|;}$
      out := next::out;
   i := i+1;
\end{lstlisting}
\caption{The SmartSum algorithm.}
\label{fig:smartsum}
\end{figure}

Assume for two adjacent databases, at most one query answer differs, and for
that query, its distance is at most one (this adjacency assumption is provided
as the precondition in function signature). Hence, for queries that generate
the same answer on adjacent databases, no privacy cost is paid. However,
privacy cost is paid twice to hide the query answers that differ: when the
noisy sum for the block containing that query is computed, and when the noisy
version of that query is used.  Hence informally, the SmartSum algorithm
satisfies $2\priv$-privacy where $\priv$ is a function parameter.  

\paragraph{Verification using \lang} \lang successfully infers the type
annotations shown in the box under function signature in
Figure~\ref{fig:smartsum}. Since all type variables are only involved in
equality constraints, only one solution exists. The transformed program is
shown at the bottom of Figure~\ref{fig:smartsum}.

By Theorem~\ref{thm:privacy}, to prove SmartSum is $2\priv$-private, it is
sufficient to verify that the postcondition $\vpriv≤2\priv$ holds for the
transformed program. We notice that this program maintains the loop invariant
shown in Figure~\ref{fig:smartsum}. One observation is that once the privacy
cost or the distance of variable $\cod{sum}$ gets positive, the query that
generates different answers must have been handled already. Hence, rest queries
must have identical answers on adjacent databases
($(\vpriv+\distance{\text{sum}}>0⇒∀j≥i.~\distance{q}[j]=0)$). Using the loop
invariant, we formally verified the desired postcondition $\vpriv≤2\priv$ using
Dafny.

\subsection{Categorical Outputs}
\label{sec:categorical}

\begin{figure}
\small
\algrule
\funsigfour{PrivBernoulli}{t$\annotation{:\tyreal_*}$}{b\annotation{:\bool}}{$0≤t≤1 ∧ 0≤t+\distance{t}≤1$}
\algrule
\begin{lstlisting}[frame=none]
$η := \uniform$;
if ($η ≤ t$) then 
  b := $\true$;
else
  b := $\false$;
\end{lstlisting}
The transformed program where underlined commands are added by the type system:
\begin{lstlisting}[frame=none]
$\vpriv$ :=0;
$\instrument{\havoc{η}; \vpriv = \vpriv - \log (1+(η≤t)?(\distance{t}≥0?0:(\distance{t}/t))}$$\Suppressnumber$
                           $\instrument{:(\distance{t}≤0?0:(\distance{t}/t)))}$;$\Reactivatenumber$
if ($η ≤ t$) then 
  b := $\true$;
else 
  b := $\false$;
\end{lstlisting}
\caption{The PrivBernoulli algorithm.}
\label{fig:privtree}
\end{figure}

Until now, we have used the Laplace mechanism, which generates numerical
outputs, as the primary randomization tool for ensuring differential privacy.
It might seem that categorical attributes would require completely different
techniques, but indeed, they can be cleanly incorporated into \lang with a new
typing rule. We briefly show how this can be done by considering a simple
mechanism that takes a private-data-dependent probability $t$ and outputs
$\true$ with probability $t$ and $\false$ with probability $1-t$. The algorithm
shown in Figure \ref{fig:privtree}.

The standard trick of generating an output $\true$ with probability $t$ can be
done by generating an uniform [0,1] random variable $x$ and returning $\true$
if $x\leq t$, and $\false$ otherwise. This trick converts numerical randomness
into categorical randomness with a notion of distance that can be aligned
between executions under related databases.
Generalizations to a larger output domain are routine and, in this way, can
allow some instantiations of the exponential mechanism
\cite{exponentialMechanism}.

To calculate the privacy cost of aligning the binary output, we need to add a
single typing rule to capture the property of uniform [0,1] distribution:
\begin{mathpar}
\inferrule*[]{ Γ(η) = \tyreal_{η\cdot \dexpr}\and -1<\dexpr≤0}
                           {Γ \proves η := \uniform \transform \cod{havoc}_{[0,1]}{η}; \vpriv = \vpriv-\log (\dexpr+1)}
\end{mathpar}

This rule requires that the random sample is aligned by a distance of $η\cdot
\dexpr$ for some $\dexpr$ (i.e., we map $η$ to $(\dexpr+1)η$ in the randomness
alignment). Easy to check this mapping is injective. By property of uniform
distribution, the privacy cost of any such assignment is $-\log(\dexpr+1)$
where $-1≤\dexpr≤0$.

To integrate this typing rule and uniform distribution into \lang, we need to
establish that: 1) the faithfulness of the transformation, and 2) the uniform
distribution satisfies Lemma~\ref{lem:pointwise}. The former is easy to check,
and we establish the latter in 
\ifreport
the appendix.
\else
the full version of this paper~\cite{report}.
\fi

With this new typing rule for uniform distribution, we can precisely compute
the privacy cost of the algorithm in Figure~\ref{fig:privtree} by providing the
following type for $η$: $Γ(η)=η\cdot\dexpr$ where
\begin{align*}
\dexpr=(η≤t)?&(\distance{t}≥0?0:(\distance{t}/t)) \\
         :&(\distance{t}≤0?0:(\distance{t}/t)) 
\end{align*}
During type checking,  rule~\ruleref{T-ODot} checks the following constraint
for the branch condition $η≤t ⇔ η+η\cdot\dexpr≤t+\distance{t}$,
which can be discharged by a SMT solver. Hence, the algorithm is transformed to
the program at the bottom of Figure~\ref{fig:privtree}. By the fact that the
newly added random source and typing rules satisfies Lemma~\ref{lem:pointwise},
the privacy cost of this subtle example is provably bounded by the transformed
cost formula in the transformed program~\footnote{We note that without \lang,
the precise calculation of privacy cost is very difficult and error-prone. To
show that the randomness alignment cancels out the difference in the
private-data-dependent probability $t$, we need to analyze four cases.  When
outputting $\true$ and $\distance{t}≥0$, the related execution must output
$\true$ as well ($η≤t∧\distance{t}≥0⇒η≤t+\distance{t}$).  When outputting
$\true$ and $\distance{t}<0$, this alignment maps $η$ to
$η(\dexpr+1)=η((t+\distance{t})/t)$. Hence, $\true$ is the output in the
related execution ($η≤t⇒η(t+\distance{t})/t≤t+\distance{t}$). Similar reasoning
applies to the case outputting  $\false$ too. Moreover, connecting this
alignment to $\priv$-privacy require is even more daunting by a
paper-and-pencil proof.}.

%% file: related.tex
\paragraph{Type systems for differential privacy}
Fuzz~\cite{Fuzz} and its successor DFuzz \cite{DFuzz} reason about the
\emph{sensitivity} (i.e., how much does a function magnify distances between
inputs) of a program. DFuzz combines linear indexed types and lightweight
dependent types to allow rich sensitivity analysis. However, those systems rely
on (without verify) external mechanisms (e.g., Laplace mechanism, Sparse Vector
method) as trusted black boxes to release final query answers, without
verifying those black boxes. \lang, on the other hand, verifies sophisticated privacy-preserving
mechanisms that releases those final answers. Sensitivity
inference~\cite{Antoni13} was proposed in the context of Fuzz. While
sensitivity inference shares the same goal of minimizing type annotation and it
also uses SMT solvers, the very different type system in \lang brings unique
challenges (Section~\ref{sec:challenges}) that do not present in Fuzz.

HOARe$^2$~\cite{Barthe15} and its
extension PrivInfer~\cite{privinfer} have the ability to relate a pair of
expressions via relational assertions that appear as refinements in types.
Hence, they can verify mechanisms that privately release final query answers as
well as private Bayesian inference algorithms. However, HOARe$^2$ and PrivInfer
incur heavy annotation burden on programmers. Moreover, they can not deal with
privacy-preserving algorithms that go beyond the composition theorem (e.g., the
Sparse Vector method). 

\paragraph{Program logic for differential privacy}
Probabilistic relational program
logic~\cite{Barthe12,EasyCrypt,BartheICALP2013,Barthe16,BartheCCS16} use custom relational
logics to verify differential privacy. These systems have successfully verified
privacy for many advanced examples. However, only the very recent work
by~\citet{Barthe16,BartheCCS16} can verify the Sparse Vector method. While
these logics are expressive enough to prove $(\priv,\delta)$ privacy, the main
difficulty with these approaches is that they use custom and complex logics
that incurs steep learning curve and heavy annotation burden.  Moreover, ad hoc
rules for loops are needed for many advanced examples.

The work by~\citet{Barthe14} transforms a probabilistic relational program to a
nondeterministic program, where standard Hoare logic can be used to reason
about privacy. However, the fundamental difference between that work and \lang
is that the former cannot verify sophisticated algorithms where the composition
theorem falls short, since it lacks the power to express subtle dependency
between privacy cost and memory state. Moreover, beneath the surface, that work
and \lang are built on very different principals and proof techniques.
Further, their approach requires heavier annotation burden since both
relational and functional (e.g., bounding privacy cost) properties are reasoned
about in the transformed program, while the former is completely and
automatically handled by the type system of \lang.

The notion of aligning randomness has been used in the recent coupling
method~\cite{Barthe16,BartheCCS16}. While the coupling method is capable of
proving $(\priv,δ)$ privacy and it does not require the injective assumption on
the alignment, the cost of doing so is the steep learning curve and heavy
annotation burden. Technically, the coupling method reasons about privacy for
each possible output (or a set of outputs), while the alignment-based theory
used in this paper aligns two program executions that will produce the same
results. The theory in this paper gives a simple proof, a light-weight type
system, and clear insight behind the type system.

\paragraph{Other language-based methods for differential privacy}
Several dynamic tools exist for enforcing differential privacy.
PINQ~\cite{pinq} tracks (at runtime) the privacy budget consumption, and
terminates the computation when the privacy budget is exhausted.
Airavat~\cite{Roy10Airavat} is a MapReduce-based system with a runtime monitor
that enforces privacy policies controlled by data providers. Recent work
by~\citet{EbadiPOPL2015} proposed Personalised Differential Privacy (PDP),
where each individual has its own personal privacy level and a dynamic system
that implements PDP. There are also methods based on computing bisimulations
families for probabilistic automata~\cite{Tschantz11,Xu2014}. However, none of
these techniques has the expressive power to provide a tight privacy cost bound
for sophisticated privacy-preserving algorithms.

%% file: conclusions.tex
The increased usage and deployment of differentially private algorithms underscores the need for formal verification methods to ensure that personal information is not leaked due to mistakes or carelessness. The ability to verify subtle algorithms should be coupled with the ability to infer most of the proofs of correctness to reduce the programmer burden during the development and subsequent maintenance of a privacy-preserving code base.

In this paper, we present a language with a lightweight type system that allows us to separate privacy computation from the alignment of random variables in hypothetical executions under related databases. Thus enabling inference and search for proofs with the minimal privacy costs.

These techniques allow us to verify (with much fewer annotations) algorithms that were out of reach of the state of the art until recently. However, additional extensions are possible. The first challenge is to extend these methods to algorithms that use hidden private state to reduce privacy costs. One example is the noisy max algorithm that adds noise to each query and returns the index of the query with the largest noisy answer (although all noisy answers are used in this computation, the fact that their values are kept secret allows more refined reasoning to replace the composition theorem). The second challenge is verifying subtle algorithms such as PrivTree \cite{PrivTree}, in which intermediate privacy costs depend on the data (hence cannot be released) but their sum can be bounded in a data-independent way. This is another case where the composition theorem can fail since it requires data-independent privacy costs. Lastly, \lang currently only verifies $\priv$-privacy, which has a nice point-wise property. We leave extending \lang to $(\priv,\delta)$-privacy as future work.